\documentclass[acmtocl]{acmtrans2m-hack} % pour que le nom du journal n'apparaisse pas sur arxiv

\usepackage[latin1]{inputenc}
\usepackage{amssymb}
\usepackage{amsmath}
\usepackage{graphicx}
\usepackage{proof}
\usepackage{color}
\usepackage{enumerate}
\usepackage{url}

\newenvironment{disparray}[2][1]%
 {\def\arraystretch{#1}\everymath{\displaystyle\everymath{}}\array{#2}}%
 {\endarray}

\newcommand{\E}{\mathsf{E}}
\newcommand{\mydownarrow}{\mathord{\downarrow}}

\newcommand{\ded}{\ensuremath{\rhd}}

\newcommand{\qded}{\ensuremath{\bowtie}} % => \eqq ??

\newcommand{\downR}{\ensuremath{\mydownarrow_{\mathcal{R}}}}
\newcommand{\rR}{\ensuremath{\to_{\cal R}}} % should be \toR ??

\newcommand{\pub}{\mathsf{pub}}
\newcommand{\F}{\ensuremath{\mathcal{F}}}
\newcommand{\Fpub}{\ensuremath{\mathcal{F}_{\pub}}}

\newcommand{\R}{\ensuremath{\mathcal{R}}}

\newcommand{\W}{\ensuremath{\mathcal{W}}}
\newcommand{\X}{\ensuremath{\mathcal{X}}}
\newcommand{\ar}{\operatorname{ar}}
\newcommand{\var}{\operatorname{var}}

\newcommand{\param}{\operatorname{par}} % \par existe déjà...
\newcommand{\dom}{\operatorname{dom}}

\newcommand{\pos}{\operatorname{pos}}

\newcommand{\st}{\operatorname{st}}
\newcommand{\stext}{\st_{\mathsf{ext}}}

\newcommand{\im}{\operatorname{im}}

\newcommand{\eq}{\operatorname{eq}}

\newcommand{\cst}[1]{\mathsf{#1}}
\newcommand{\w}{\mathsf{w}}
\newcommand{\symbfun}[1]{\mathsf{#1}}
\newcommand{\f}{\symbfun{f}}
\newcommand{\g}{\symbfun{g}}
\newcommand{\h}{\symbfun{h}}
\newcommand{\mal}{\symbfun{mal}}
\newcommand{\dec}{\symbfun{dec}}
\newcommand{\enc}{\symbfun{enc}}
\newcommand{\proj}{\symbfun{proj}}

\newcommand{\checksign}{\mathsf{checksign}}
\newcommand{\sign}{\mathsf{sign}}
\newcommand{\blind}{\mathsf{blind}}
\newcommand{\unblind}{\mathsf{unblind}}
\newcommand{\okay}{\mathsf{ok}}

\newcommand{\Init}{\operatorname{Init}}
\newcommand{\Ctx}{\operatorname{Ctx}}

\newcommand{\homo}{\mathsf{hom}}
\newcommand{\pref}{\mathsf{pref}}
\newcommand{\add}{\mathsf{add}}

\newcommand{\Longrightdotarrow}{\stackrel{%
    \mbox{\raisebox{0pt}[0pt]{$\centerdot$}}}{\Longrightarrow}}

\newtheorem{theorem}{Theorem}[section]

\newtheorem{corollary}[theorem]{Corollary}
\newtheorem{proposition}[theorem]{Proposition}
\newtheorem{lemma}[theorem]{Lemma}
\newdef{definition}[theorem]{Definition}
\newdef{remark}[theorem]{Remark}
\newdef{example}[theorem]{Example}

\newcommand{\BibTeX}{{\rm B\kern-.05em{\sc i\kern-.025em b}\kern-.08em
    T\kern-.1667em\lower.7ex\hbox{E}\kern-.125emX}}

\title{YAPA: {A generic tool for computing intruder
    knowledge}}
\author{Mathieu Baudet\\ MLstate, France\\
  Véronique Cortier\\LORIA - CNRS, France\\
  Stéphanie Delaune\\ LSV, ENS Cachan \& CNRS \& INRIA Saclay Île-de-France, France }
\begin{abstract}
Reasoning about the knowledge of an attacker is a necessary step in many 
formal analyses of security protocols. In the framework of the applied 
pi calculus, as in similar languages based on equational logics, 
knowledge is typically expressed by two relations: deducibility and 
static equivalence. Several decision procedures have been proposed for 
these relations under a variety of equational theories. However, each 
theory has its particular algorithm, and none has 
been implemented so far.

\smallskip{}

We provide a generic procedure for deducibility and static 
equivalence that takes as input any convergent rewrite system.
We show that our algorithm covers
most of the existing decision procedures for convergent theories. We also 
provide an efficient implementation, and compare it 
briefly with the tools ProVerif and KiSs.
\end{abstract}

\category{}{}{}%

\terms{}


\begin{document}

\begin{bottomstuff}
Author's address: S. Delaune, Laboratoire Sp\'ecification \& V\'erification -
61, avenue du pr\'esident Wilson - 94 230 Cachan.\newline
This work has been partly supported by the ANR-07-SESU-002 AVOTÉ.
A large part of it was done while the first author was working at the ANSSI.
\end{bottomstuff}
\maketitle

\section{Introduction}
\label{sec:intro}

%
% Contexte
%

Understanding security protocols often requires reasoning about the
information accessible to an on-line attacker. Accordingly, many formal
approaches to security rely on a notion of
\emph{deducibility}~\cite{lowe96breaking,MS02} that models whether
a piece of data, typically a secret, is retrievable from a finite set
of messages.
Deducibility, however, does not always suffice to reflect the
knowledge of an attacker. Consider for instance a protocol sending an
encrypted Boolean value, say, a vote in an electronic voting
protocol. Rather than deducibility, the key idea to express
confidentiality of the plaintext is that an attacker should
not be able to \emph{distinguish} between the sequences of messages
corresponding to each possible value. (Such security considerations 
typically motivate the
 use of randomized encryption.)

\medskip{}

%
% Deducibility and Static Equivalence
%

In the framework of the applied pi-calculus~\cite{AbadiFournet2001},
as in similar languages based on equational
logics~\cite{BlanchetAbadiFournetJLAP07},
indistinguishability corresponds to a relation called \emph{static
  equivalence}: roughly, two sequences of messages are \emph{statically equivalent}
when they satisfy the same algebraic relations from the attacker's
point of view.
Static equivalence plays an important role in the study of guessing
attacks (e.g.~\cite{Corin_Doumen_Etalle_WISP04_off_line_guessing_attacks,baudet-ccs2005,AbadiBW06}), as well as for anonymity
properties and electronic voting protocols (e.g.~\cite{DKR-jcs08}).
Static equivalence is also used for specifying privacy in the context of RFID protocols~\cite{MyrtoRFID09}.
%\marginpar{S.D: autre application untraceability in RFID protocols (papier
%  Myrto RISC'09) ?}
%
%Most often, the knowledge of the attacker is described in terms of
%deducibility~\cite{lowe96breaking,MS02}.
%Given a set of messages~$\phi$ representing the
%knowledge of the attacker and another message~$M$, intuitively the secret, one
%can ask
%whether an attacker is able to compute~$M$ from~$\phi$, using his deducibility capabilities. For instance,
%he may encrypt and decrypt when he knows the corresponding key.
%The concept of deducibility does not always suffice for expressing the
%knowledge of an attacker. For example, if we consider a protocol that
%transmits an encrypted Boolean value (\emph{e.g.} the value of a
%vote), we may ask whether an attacker can learn which value has been
%sent. 
%Of course, the attacker already knows the Boolean true and false.
%We need to express the fact that the two versions of the protocol, one
%running with the Boolean value true and the other one with false are
%\emph{indistinguishable}.
%%Besides allowing more careful formalization of secrecy properties,
%
%
In several cases, this notion has also been shown to imply the
more complex and precise notion of cryptographic
indistinguishability~\cite{BCK-ICALP2005,AbadiBW06}, related to probabilistic polynomial-time Turing machines.
Two sequences of messages are \emph{cryptographically
  indistinguishable} when their corresponding bit-string implementations are
indistinguishable to any 
probabilistic polynomial-time Turing machine.

We emphasize that both deducibility and static equivalence apply to
observations on finite sets of messages, and do not take into account
the dynamic behavior of protocols. (This justifies the expression \emph{static equivalence}.)
Nevertheless, deducibility is used as a subroutine by many general
decision procedures~\cite{CLS03,RT03xor}.
Besides, it has been shown that observational equivalence in the
applied pi-calculus coincides with labeled bisimulation~\cite{AbadiFournet2001}, that is,
corresponds to checking a number of static equivalences and some standard
bisimulation conditions.

\medskip{}

%
% Related work
%

Deducibility and static equivalence rely on an underlying equational
theory for axiomatizing the properties of cryptographic
functions. Many decision
procedures \cite{AbadiCortierTCS06,CortierDelaune-LPAR07-monoidal}
have been proposed to compute these relations under a variety of
equational theories, including symmetric and asymmetric encryptions,
signatures, exclusive OR, and homomorphic operators.
However, except for the class of subterm convergent
theories~\cite{AbadiCortierTCS06}, which covers the standard flavors of
encryption and signature, each of these decision results introduces
a new procedure, devoted to a particular theory. Even in the case of
the general decidability criterion given in~\cite{AbadiCortierTCS06},
we note that the algorithm underlying the proof has to be
adapted for each theory, depending on how the criterion is fulfilled.

Perhaps as a consequence of this fact, none of these decision
procedures has been implemented so far. When we began this work, 
the only tool able to verify static equivalence was
ProVerif~\cite{BlanchetCSFW01,BlanchetAbadiFournetJLAP07}. This general tool can
handle various equational theories and analyze security protocols
under active adversaries. However termination of the verifier is not
guaranteed in general, and protocols are subject to (safe)
approximations. Since then, a new tool, called KiSs, has been developed~\cite{CDK-cade09}. The
procedure implemented in KiSs has many concepts in common with %the one presented in
 a preliminary version of this work~\cite{BCD-RTA09} but targets a different class of equational theories.
%
% For some equational theories, we have observed that some features of ProVerif
% specific to the active case could cause failure even though the
% problems of static equivalence under consideration were decidable.

\medskip{}

The present work aims to fill this gap between theory and
implementation and propose an efficient tool for deciding deducibility
and static equivalence in a uniform way. It is initially inspired from
a procedure for solving more general constraint systems related to
active adversaries and equivalence of finite processes, presented
in~\cite{baudet-ccs2005}, with corrected extended version
in~\cite{THESE-baudet07} (in French). However, due to the complexity
of the constraint systems, this decision procedure was only studied for
subterm convergent theories, and remains too complex to enable an
efficient implementation.

%
% Our contributions
%

\paragraph*{Our Contributions}
In this paper, we provide and study a generic
procedure for checking deducibility and static equivalence, taking as
input any convergent theory (that is, any equational theory described
by a finite convergent rewrite system). We prove the algorithm sound
and complete, up to explicit failure cases.
Note that (unfailing) termination cannot be guaranteed in general since the
problem of checking deducibility and static equivalence is
undecidable, even for convergent theories~\cite{AbadiCortierTCS06}.
To address this issue and turn our algorithm into a decision
procedure for a given convergent theory, we provide two criteria.
First, we define a syntactic criterion on the rewrite rules that
ensures that the algorithm never fails. This criterion is enjoyed in
particular by any convergent subterm theory, as well as the theories
of blind signature and homomorphic encryption. Termination often
follows from a simple analysis of the rules of the algorithm: as a
proof of concept, we obtain a new decidability result for deducibility
and static equivalence for the prefix theory, representing encryption
in CBC mode. 
%% The prefix
%% theory is the theory of encryption enriched with the equation 
%% $\pref(\enc(\langle x, y \rangle, z)) = \enc(x,z)$ that represents the
%% ability of retrieving the encrypted prefix of a cypher-text. This
%% happens in particular when using the  CBC mode. 
Second, we provide a termination criterion based on deducibility:
provided that failure cannot occur, termination on a given input is
equivalent to the existence of some natural finite representation of deducible
terms.
As a consequence, we obtain that our algorithm can decide
deducibility and static equivalence for all the convergent theories
shown  to be decidable in~\cite{AbadiCortierTCS06}.

Our second contribution is an efficient implementation of this generic
procedure, called YAPA.
After describing the main features of the implementation, we report
several experiments suggesting that our tool computes static
equivalence faster and for more convergent theories than the general
tool ProVerif~\cite{BlanchetCSFW01,BlanchetAbadiFournetJLAP07}.
We also outline the main differences between YAPA and the recent tool Kiss.

%
% Outline
%

\paragraph*{Outline}
We introduce our setting in Section~\ref{sec:prelim}, in particular the notion of term algebra and  equational theory, that are used to model cryptographic primitives.
Deducibility and static equivalence are defined in Section~\ref{sec:problem}.
We describe our procedure in Section~\ref{sec:rules} and prove its correctness and completeness in Section~\ref{sec:completeness}.
We provide criteria for preventing failure in Section~\ref{sec:fail} and for ensuring termination in Section~\ref{sec:termination}. The implementation of our procedure is discussed in Section~\ref{sec:experiment}.
Some concluding remarks and perspectives can be found in Section~\ref{sec:conclu}.
A number of technical proofs have been postponed to the appendix to ease the presentation.

%%% Local Variables: 
%%% mode: latex
%%% TeX-master: "main"
%%% End: 

\section{Preliminaries}
\label{sec:prelim}

\subsection{Term algebra}
\label{subsec:term}

We start by introducing the necessary notions to describe
cryptographic messages in a symbolical way.
For modeling cryptographic primitives, we assume given a set of
\emph{function symbols}~$\F$ together with an arity 
function~${\ar: \F \to \mathbb{N}}$. Symbols in~$\F$ of arity~$0$ 
are called \emph{constants}.
We consider a set of \emph{variables}~$\X$ and a set of additional
constants~$\W$ called $\emph{parameters}$.
The (usual, first-order) term algebra generated by~$\F$ over~$\W$ and~$\X$ 
is written $\F[\W \cup \X]$ with elements denoted by $T, U, T_1
\ldots$ More generally, we write $\F'[A]$ for the least set of terms
containing a set $A$ and stable by application of symbols
in $\F' \subseteq \F$.

We write~$\var(T)$ (resp.~$\param(T)$) for the set of variables
(resp. parameters) that occur in a term~$T$. These notations are
extended to tuples and sets of terms in the usual way.
The set of positions  of a term~$T$ is
written~$\pos(T) \subseteq \mathbb{N}^*$, and its set of subterms~$\st(T)$.  
The subterm of~$T$ at position~${p \in \pos(T)}$ is written~$T|_p$.  The
term obtained by replacing~$T|_p$ with a term~$U$ in~$T$ is
denoted~$T[U]_p$.
% The set of subterms of~$T$ is written $\st(T)$.

\medskip{}

A \emph{(finite, partial) substitution}~$\sigma$ is a mapping from a
finite subset of variables%~$\X$
, called its \emph{domain} and written~$\dom(\sigma)$, to terms. The
\emph{image} of a substitution is its image as a mapping
$\im(\sigma)=\{\sigma(x)~|~x\in \dom(\sigma)\}$.
Substitutions are extended to endomorphisms of $\F[\X \cup \W]$ as
usual. We use a postfix notation for their application.
A term~$T$ (resp. a substitution~$\sigma$) is \emph{ground} if
$\var(T)=\emptyset$ (resp. $\var(\im(\sigma))=\emptyset$).

\medskip{}

For our cryptographic purposes, it is useful to distinguish a subset~$\Fpub$ 
of~$\F$, made of \emph{public function symbols}, that is,
intuitively, the symbols made available to the attacker.
A \emph{recipe} (or \emph{second-order term}) $M$, $N$, $M_1$\ldots{}
is a term in ${\Fpub[\W \cup \X]}$, that is, a term containing no
\emph{private} (non-public) function symbols.
A \emph{plain term} (or \emph{first-order term}) $t$, $r$, $s$,
$t_1$\ldots{} is a term in $\F[\X]$, that is, containing no
parameters.
A \emph{(public, ground, non-necessarily linear) $n$-ary context}~$C$ 
is a recipe in $\Fpub[\w_1,\ldots, \w_n]$, where
we assume a fixed countable subset of parameters $\{\w_1,\ldots,
\w_n, \ldots\} \subseteq \W$. If~$C$ is a $n$-ary context, $C[T_1, \ldots,
  T_n]$ denotes the term obtained by replacing each
occurrence of $\w_i$ with~$T_i$ in~$C$.

\subsection{Rewriting}
\label{subsec:trs}

A \emph{rewrite system}~$\R$ is a finite set of \emph{rewrite rules}~${l \to
  r}$ where~${l,r \in \F[\X]}$ and such that~${\var(r) \subseteq \var(l)}$. A
term~$S$ \emph{rewrites} to~$T$ by~$\R$, denoted~${S \to_\R T}$, if there exist 
${l \to r}$ in~$\R$, ${p \in \pos(S)}$ and a substitution~$\sigma$ such that
${S|_p = l\sigma}$ and ${T = S[r\sigma]_p}$.
We write $\to^+_{\R}$ for the transitive closure of $\to_{\R}$,
$\to^*_{\R}$ for its reflexive and transitive closure,
and $=_{\R}$ for its reflexive, symmetric and transitive closure.

\medskip{}

\noindent A rewrite system~$\R$ is \emph{convergent} if it is:
\begin{itemize}
\item  \emph{terminating}, i.e. there is no infinite chains
$T_1 \to_\R T_2 \to_\R \ldots$; and
\item \emph{confluent}, i.e. for every terms $S$, $T$ such
that $S =_{\R} T$, there exists $U$ such that $S
\to^*_{\R} U$ and  $T \to^*_{\R} U$.
\end{itemize}

A term $T$ is \emph{$\R$-reduced} if there is no term~$S$ such that
$T \to_{\R} S$.
If $T \to^{*}_{\R} S$ and~$S$ is  $\R$-reduced
then $S$ is \emph{a $\R$-reduced form of $T$}. When this reduced form
is unique (in particular if $\R$ is convergent), we write $S = T\downR$ 
(or simply $T\mydownarrow$ when $\R$ is clear from the context).
%
% $T\downR$ denotes a fixed $\R$-reduced form of $T$, when it exists.

\subsection{Equational theories}
\label{sec:eq}
We equip the signature~$\F$ with an equational theory represented by 
a set of equations~$\mathcal{E}$ of the form $s=t$ with $s,t \in \F[\X]$.
The equational theory~$\E$ generated by~$\mathcal{E}$ is the least set of
equations containing~$\mathcal{E}$ that is stable under the axioms of
congruence (reflexivity, symmetry, transitivity, application of
function symbols) and under application of substitutions. We
write~$=_\E$ for the corresponding relation on terms.
%
% 
% By Birkhoff's theorem, $=_\E$ is also the reflexive, symmetric,
% transitive, closure of the relation $\rightarrow_E$, defined by $S
% \rightarrow_E T$ iff there exist $s=t$ in $\cal E$, a substitution
% $\sigma$ of domaine $\var(S,T)$ and $p \in \pos(S)$ such that $S|p =
% s\sigma$ and $T = S[s\sigma]_p$.
%% where $E$ is seen as a set of rewrite rules is not enough because 
Equational theories 
have proved very useful for modeling algebraic properties of
cryptographic primitives (see e.g.~\cite{CDL05-survey} for a survey).
%~\cite{CDL05-survey,AbadiCortierTCS06}. 

\medskip{}

We are particularly interested in theories~$\E$ that
can be represented by a convergent rewrite system~$\R$,
i.e. theories for which there exists a convergent rewrite
system~$\R$ such that the two relations $=_{\R}$ and $=_{\E}$ coincide. The
rewrite system~$\R$ ---and by extension the equational theory $\E$--- is
\emph{weakly subterm convergent} if, in addition, we have that for every rule
${l \to r \in \R}$, $r$ is either a subterm of~$l$ or a ground
$\R$-reduced term. This class encompasses the 
class of subterm convergent theories
%one of the same name
  used in~\cite{AbadiCortierTCS06} (for every rule $l \to r \in \R$, $r$ is a
  subterm of~$l$ or a constant), the class of dwindling theories used
  in~\cite{ANR07}, and the class of public-collapsing theories
  introduced in~\cite{dj-ccs-2004}.

\begin{example}
\label{ex:theory}
Consider
the signature $\F_\enc =  \{\dec, \enc, \langle \_, \_\rangle, \proj_1,
\proj_2\}$.
The symbols~$\dec, \enc$ and $\langle \_, \_\rangle$ are functional symbols of arity 2 that
represent respectively the decryption, encryption and pairing functions, whereas
$\proj_1$ and~$\proj_2$ are functional symbols of arity 1 that
represent the projection function on the first and the
second component of a pair, respectively. 
The equational theory of pairing and symmetric (deterministic) encryption,
denoted by~$\E_\enc$, is
generated by the equations 
\[
{\cal E}_\enc = \{\dec(\enc(x,y),y) = x, \;\; \proj_1(\langle x,y \rangle ) =
x, \;\;\proj_2(\langle x,y \rangle ) = y\}.
\]

\medskip{}

Motivated by the modeling of the ECB mode of encryption, %% ~\cite{Menezes_Oorschot_Vanstone_HAC96}
we may also consider an encryption symbol that is homomorphic with respect to pairing:
\[
{\cal E}_{\homo} = {\cal E}_\enc \cup \left\{
\begin{array}{rcl}
\enc(\langle x, y \rangle, z) & =& \langle \enc(x,z), \enc(y,z) \rangle\\
\dec(\langle x, y \rangle, z) & = & \langle \dec(x, z), \dec(y,z) \rangle
\end{array} 
\right\}.
\]
If we orient the equations from left to right, we obtain two rewrite
systems~$\R_\enc$ 
and~$\R_\homo$. Both rewrite systems are convergent, only~$\R_\enc$
is (weakly) subterm convergent.
Other examples of subterm convergent theories can be found in~\cite{AbadiCortierTCS06}.
%Mathieu: ici on pourrait espérer l'intérêt du ``weakly'' pour modéliser les fonctions partielles.
\end{example}

From now on, we assume given a equational theory~$\E$ represented by 
a convergent rewrite system~$\R$.
A symbol~$f$ is \emph{free} if~$f$ does
not occur in~$\R$.
In order to model (an unbounded number of) random values possibly
generated by the attacker, we assume that~$\Fpub$ contains
infinitely many free public constants.
We will use free private constants to model
secrets, for instance the secret keys used to encrypt a message.
Private (resp. public) free constants are closely related to bound
(resp. free) \emph{names} in the framework of the applied pi
calculus~\cite{AbadiFournet2001}.
Our formalism also allows one to consider non-constant private
symbols.

%%% Local Variables: 
%%% mode: latex
%%% TeX-master: "main"
%%% End: 

\section{Deducibility and static equivalence}
\label{sec:problem}
In order to describe the cryptographic messages observed or inferred by an
attacker, we introduce the following notions of deduction facts and
frames.

\medskip{}

A \emph{deduction fact} is a pair, written $M \ded t$, made of a
recipe $M \in \Fpub[\W \cup \X]$ and a plain term $t \in \F[\X]$. Such a
deduction fact is \emph{ground} if $\var(M,t)=\emptyset$.
A \emph{frame}, denoted by letters $\varphi$, $\Phi$,
$\Phi_0$\ldots{}, is a finite set of ground deduction facts.
The \emph{image} of a frame is defined by $\im(\Phi) = \{t\mid M\ded
t\in\Phi\}$.  A frame $\Phi$ is \emph{one-to-one} if $M_1\ded t$,
$M_2\ded t\in\Phi$ implies $M_1=M_2$.

\medskip{}

A frame~$\varphi$ is \emph{initial} if it is of the form
$\varphi = \{w_1 \ded t_1, \ldots, w_\ell \ded t_\ell\}$ 
for some distinct parameters $w_1$, \ldots, $w_\ell \in \W$.
The parameters $w_i$ can be
seen as labels that refer to the messages observed by an attacker.
Initial frames are closely related to the notion of frames in the
applied pi-calculus~\cite{AbadiFournet2001}. 
The only difference is that, in initial frames,  values initially unknown to an attacker are modeled 
by private constants while they are modeled by \emph{restricted names} in the applied pi-calculus.
%Mathieu: à voir si vous aimez
Name generation and binding are important features of the (general) applied
calculus but are unessential when considering 
finite processes, and in particular frames.
Given such an initial frame~$\varphi$, we denote by $\dom(\varphi)$ its
\emph{domain} $\dom(\varphi) = \{w_1, \ldots, w_\ell\}$. If
$\param(M)\subseteq \dom(\varphi)$, we write $M \varphi$ for the term
obtained by replacing each~$w_i$ by~$t_i$ in $M$.
 We note that if 
in addition $M$ is ground then $t= M \varphi$ is a ground plain term.

\subsection{Deducibility, recipes}
\label{subsec:deduction}

Classically (see e.g.~\cite{AbadiCortierTCS06}), a ground
term~$t$ is \emph{deducible} modulo~$\E$ from an initial
frame~$\varphi$, written $\varphi \vdash_\E t$,
if there exists $M \in
\Fpub[\dom(\varphi)]$ such that $M \varphi =_\E t$. This corresponds
to the intuition that the attacker may compute (infer)~$t$ from~$\varphi$.
For the purpose of our study, we generalize this notion to
arbitrary (i.e. non-necessarily initial)
frames, and even
sets of (non-necessarily ground) 
deduction facts $\phi$, using the
notations $\ded_\phi$ and $\ded^\E_\phi$  defined as follows.

\begin{definition}[Deducibility]
Let $\phi$ be finite set of deductions facts. %, for instance a frame.
We say that \emph{$M$ is a recipe of~$t$ in~$\phi$}, written
$M \ded_\phi t$, if there exist a 
(public, ground, non-necessarily linear) 
$n$-ary context $C$ and some deduction facts $M_1 \ded t_1$,
\ldots, $M_n \ded t_n$ in $\phi$ such that
$M = C[M_1, \ldots, M_n]$ and $t = C[t_1, \ldots, t_n]$. 
In that case, we say that $t$ is \emph{syntactically deducible} from~$\phi$, also written $\phi \vdash t$.

We say that \emph{$M$ is a recipe of $t$ in $\phi$ modulo $\E$}, written
$M \ded^\E_\phi t$, if there exists a term $t'$ such that
$M \ded_\phi t'$ and $t' =_\E t$.
In that case, we say that $t$ is \emph{deducible from~$\phi$ modulo $\E$}, written $\phi \vdash_\E t$.
\end{definition}

We note that $M \ded_\varphi t$ is equivalent to $M\varphi = t$ when
$\varphi$ is an initial frame and when~$t$ (or equivalently $M$) is ground.
We also note that in the case of a frame~$\varphi$, since our
contexts~$C$ are ground and public, $M \ded_{\varphi} t$ implies
$\var(M,t)=\emptyset$ and $\param(M) \subseteq \param(\varphi)$.

\begin{example}
\label{ex:deducibility}
Consider the equational theory $\E_\enc$ described in
Example~\ref{ex:theory}.
Let~$\varphi_0 = \{\w_1 \ded \enc(\cst c_0, \cst k),
\w_2 \ded \cst k\}$ where $\cst c_0$ is a public constant and $\cst k$ is a
private constant.
We have that $\varphi_0$ is a set of deduction facts. Since, these facts are
ground, $\varphi_0$ is actually a frame. Moreover, this frame is initial.
We have that $\langle \w_2, \w_2 \rangle \ded_{\varphi_0} \langle \cst k, \cst k
\rangle$, $\cst c_0 \ded_{\varphi_0} \cst c_0$,
and $\dec(\w_1,\w_2) \ded_{\varphi_0}^{\E_\enc} \cst c_0$.
\end{example}

\subsection{Static equivalence, visible equations}
\label{subsec:staticequivalence}

Deducibility does not always suffice for expressing the knowledge of an
attacker. In particular, it does not account for the partial
information that an attacker may obtain about secrets.
Sometimes, the attacker can deduce exactly the same set of terms from two different
frames but he could still be able to tell the difference between these
two frames.
This issue motivates the study of visible
equations and static equivalence (see~\cite{AbadiFournet2001}), defined as follows.

\begin{definition}[Static equivalence]
  Let~$\varphi$ be an initial frame. The set of \emph{visible equations
    of~$\varphi$ modulo~$\E$} is defined as
\[
\eq_\E(\varphi) = \{ M \qded N \;|\; M,N \in \Fpub[\dom(\varphi)],\; M \varphi
=_\E N \varphi \}
\]
where $\qded$ is a dedicated commutative symbol.
Two initial frames $\varphi_1$ and $\varphi_2$ with the same domain 
are \emph{statically
  equivalent} modulo $\E$, written~${\varphi_1 \approx_{\E}
  \varphi_2}$, if their sets of visible equations are equal, i.e.
 $\eq_{\E}(\varphi_1) = \eq_{\E}(\varphi_2)$.
\end{definition}

This definition is in line with static equivalence in the applied pi
calculus~\cite{AbadiFournet2001} where bounds names would be replaced by free private constants.

\begin{example}
\label{ex:static}
Consider again the equational theory~$\E_{\enc}$ given
in Example~\ref{ex:theory}.
Let $\varphi_0 = \{ \w_1 \ded \enc(\cst c_0, \cst k), \; \w_2 \ded \cst k \}$ and
$\varphi_1 = \{ \w_1 \ded \enc(\cst c_1,\cst k), \; \w_2 \ded \cst k \}$ where
$\cst c_0$, $\cst c_1$ are
public constants and~$\cst k$ is a private constant.
We have that:
\begin{itemize}
\item $(\enc(\cst c_0,\w_2) \qded \w_1) \in \eq_{\E_{\enc}}(\varphi_0)$, and
\item  $(\enc(\cst c_0,\w_2) \qded \w_1) \not\in \eq_{\E_{\enc}}(\varphi_1)$.
\end{itemize}
Hence, $\eq_{\E_\enc}(\varphi_0) \neq \eq_{\E_\enc}(\varphi_1)$ and the
two frames $\varphi_0$ and $\varphi_1$ are not statically equivalent.
However, it can be shown that $\{ \w_1 \ded \enc(\cst c_0, \cst k)\}
\approx_{\E_\enc} \{\w_1 \ded \enc(\cst c_1, \cst k)\}$.
\end{example}

\medskip{}

For the purpose of finitely describing the set of visible
equations~$\eq_\E(\varphi)$ of an initial frame, we introduce
\emph{quantified equations} of the form $\forall z_1,\ldots, z_q. M \qded N$
where $z_1$, \ldots, $z_q \in \X$, $q \geq 0$ and $\var(M,N) \subseteq
\{z_1, \ldots, z_q\}$.
In what follows, finite sets of quantified
equations are denoted $\Psi$, $\Psi_0$,\ldots{} We write $\Psi \models
M \qded N$
%or $(M \qded N) \in \eq(\Psi)$
%, or indifferently $M \qded_\Psi N$,
when the ground equation $M \qded N$ is a consequence of $\Psi$ in the
usual, first-order logics with equality axioms for the relation
$\qded$ (that is, reflexivity, symmetry, transitivity and compatibility with
symbols in $\Fpub$).
When no confusion arises, we may refer to quantified equations simply
as \emph{equations}. As usual, quantified equations are considered up to
renaming of bound variables.

\begin{example}
\label{ex:quantifiedeq}
Consider the equational theory~$\E_{\homo}$ given
in Example~\ref{ex:theory}. Let $\varphi = \{
\w_1 \ded \enc(\langle \cst c_0, \cst c_1\rangle, \cst k), \; 
\w_2 \ded \langle \enc(\cst c_0,\cst  k), \enc(\cst c_1, \cst k) \rangle, \; 
\w_3 \ded \cst k\}$ where $\cst c_0$ and $\cst c_1$ are public constants and
$\cst k$ is a private constant. In the set $\eq_{\E_{\homo}}(\varphi)$, we have,
among others, 
$\w_1 \qded \w_2$ and
 $\dec(\w_1, M) \qded \langle \dec(\proj_1(\w_1), M), \dec(\proj_2(\w_1),M)
\rangle$ for every term $M \in \F_\pub[\dom(\varphi)]$. 
Indeed, we have that:
\[
\begin{array}{rll}
\dec(\w_1, M)\varphi &=& \dec(\enc(\langle \cst c_0, \cst c_1\rangle, \cst k),
M\varphi) \\
&=_{\E_\homo}& \langle \dec(\enc(\cst c_0,\cst k), M\varphi),  \dec(\enc(\cst
c_1,\cst k), M\varphi) \rangle\\
& =_{\E_{\homo}}& \langle \dec(\proj_1(\w_1), M),
\dec(\proj_2(\w_1),M) \rangle \varphi
\end{array}
\]

\noindent This infinite set 
will be represented with the quantified equation:
\[
\forall z. \; \dec(\w_1, z) \qded \langle \dec(\proj_1(\w_1), z), \dec(\proj_2(\w_1),z)
\rangle.
\]
\end{example}

%%% Local Variables: 
%%% mode: latex
%%% TeX-master: "main"
%%% End: 

\section{Main procedure}
\label{sec:rules}

In this section, we describe our algorithms for checking deducibility
and static equivalence on convergent rewrite systems.
After some additional notations, we present the core of the procedure,
which consists of a set of transformation rules used to saturate a
frame and a finite set of quantified equations.
The result of the saturation can be
seen as a finite description of the deducible terms and visible
equations of the initial frame under consideration.
We then show how to use this procedure to decide deducibility and
static equivalence, provided that saturation succeeds.
(Recall that
static equivalence and deduction are undecidable for convergent
theories~\cite{AbadiCortierTCS06}.)

Soundness and completeness of the saturation procedure are detailed in
Section~\ref{sec:completeness}. We provide sufficient conditions on
the rewrite systems to ensure success of saturation and termination in
Section~\ref{sec:fail} and Section~\ref{sec:termination}.

\subsection{Decompositions of rewrite rules}

Before stating the procedure, we introduce the following
notion of \emph{decomposition}
to account for the possible
superpositions of an attacker's
context (that is, a recipe in our setting)
with a left-hand side of rewrite rule.

\begin{definition}[Decomposition] \label{def:decomp}
  Let $n,p,q$ be non-negative
  integers. 
  A \emph{$(n,p,q)$-decomposition} of a term~$l$ (and by an extension
  of any rewrite rule $l \to r$) is a (public, ground,
  non-necessarily linear) context $D \in \Fpub[\W]$ such that
  $\param(D) = \{\cst w_1, \ldots, \cst w_{n+p+q}\}$ and $l = D[l_1,
  \ldots, l_n,y_1, \ldots, y_p, z_1, \ldots, z_q]$ where
\begin{itemize}
\item $l_1, \ldots, l_n$ are mutually-distinct non-variable terms,
\item $y_1, \ldots, y_p$ and $z_1, \ldots, z_q$ are mutually-distinct
  variables, and
\item $y_1, \ldots, y_p \in \var(l_1, \ldots, l_n)$ whereas $z_1, \ldots,
  z_q \not\in \var(l_1, \ldots, l_n)$.
\end{itemize}
A decomposition $D$ is \emph{proper} if it is not a parameter (i.e. $D
\neq \w_1$).
\end{definition}
 In order to avoid unnecessary computations, $(n,p,q)$-decompositions are considered up to
 permutations of parameters in the sets
 $\{\cst w_1, \ldots, \cst w_{n}\}$, $\{\cst w_{n+1}, \ldots, \cst
 w_{n+p}\}$ and $\{\cst w_{n+p+1}, \ldots, \cst w_{n+p+q}\}$ respectively.

\begin{example}
\label{ex:decompo}
  Consider the rewrite rule $\dec(\enc(x,y),y) \to x$. This rule
  admits two proper decompositions up to permutation of parameters:
\begin{itemize}
\item $D_1 = \dec(\enc(\cst w_1, \cst w_2), \cst w_2)$ where $n=0$,
  $p=0$, $q=2$, $z_1=x$, $z_2=y$;
\item $D_2 = \dec(\cst w_1, \cst w_2)$ where $n=1$,
  $p=1$, $q=0$, $l_1=\enc(x,y)$ and $y_1=y$.
\end{itemize} 

Now, consider the rewrite rule $\dec(\langle x, y \rangle, z) \to \langle
\dec(x,z), \dec(y,z) \rangle$. This rule also admits two proper
decompositions:
\begin{itemize}
\item $D_3 = \dec(\langle \w_1,\w_2 \rangle, \w_3)$ where $n = 0$, $p=0$, $q=3$,
  $z_1 = x$, $z_2 = y$, $z_3=z$;
\item $D_4 = \dec(\w_1,\w_2)$ where $n=1$, $p=0$, $q=1$, $l_1=\langle x,
  y\rangle$, $z_1 = z$.
\end{itemize}
\end{example}

\subsection{Transformation rules}\label{subsec:rules}

To check deducibility and static equivalence, 
we proceed by saturating an
initial frame, adding some deduction facts and equations satisfied by the
frame.
We consider \emph{states} that are either the failure state $\bot$ or
a couple $(\Phi,\Psi)$ formed by
a one-to-one frame $\Phi$ in
$\R$-reduced form and
a finite set of quantified equations~$\Psi$.

\medskip{}

Given an initial frame $\varphi$, our procedure starts from an
initial state associated to $\varphi$, denoted by
$\Init(\varphi)$, obtained by reducing $\varphi$ and replacing
duplicated terms by equations. Formally, $\Init(\varphi)$
is the result of a procedure recursively defined as follows:
$\Init(\emptyset)  =  (\emptyset,\emptyset)$, and assuming
$\Init(\varphi) = (\Phi,\Psi)$, we have
\[
\Init(\varphi\uplus\{w\ded t\}) =
\begin{cases}
(\Phi,\Psi\cup\{w \qded w'\})& \text{ if there exists some } w' \ded t\downR \in \Phi\\
(\Phi\cup\{w\ded t\downR\},\Psi) & \text{ otherwise.}
\end{cases}
\]

\begin{example}
Consider the frames $\varphi_0$, $\varphi_1$ and $\varphi$ introduced
respectively in Example~\ref{ex:static} and Example~\ref{ex:quantifiedeq}.
We have that $\Init(\varphi_0) = (\varphi_0,\emptyset)$, $\Init(\varphi_1) =
(\varphi_1,\emptyset)$ and $\Init(\varphi) = (\{\w_1 \ded \langle \enc(\cst c_0,
\cst k),\enc(\cst c_1,\cst k)\rangle,\w_3 \ded \cst k\}, \; \{\w_1 \qded \w_2\})$.
\end{example}
%%%%%%%%%%%%%%%%%%%%%%%%%%%%%%%%%%%%%%%%%%%%%%%%%%%%%%%%%%%%%%%
%                                                             %
%                        BEGIN FIGURE                         %
%                                                             %
%%%%%%%%%%%%%%%%%%%%%%%%%%%%%%%%%%%%%%%%%%%%%%%%%%%%%%%%%%%%%%%

\begin{figure}[t]
\begin{flushleft}

\textbf{A. Inferring deduction facts and equations by context reduction}

\smallskip{}

Assume that
\[\begin{disparray}{l}
l=D[l_1,\ldots,l_n,y_1,\ldots,y_p,z_1,\ldots,z_q] \text{ is a proper decomposition of }(l \to r) \in \R\\
M_1 \ded t_1, \ldots, M_{n+p} \ded t_{n+p} \in \Phi \\
(l_1,\ldots,l_n,y_1,\ldots,y_p)\,\sigma = (t_1, \ldots, t_{n+p})\\
\end{disparray}
\]
%Then
\begin{minipage}{\textwidth}
\begin{enumerate}%[(1)]
\item If there exists $M = \Ctx(\Phi \cup \{z_1 \ded z_1, \ldots,
  z_q \ded z_q\} \vdash^?_\R r\sigma)$, then
\begin{equation}
(\Phi,\Psi) \Longrightarrow (\Phi,\Psi \cup \{\forall z_1, \ldots,
  z_q.D[M_1, \ldots,M_{n+p}, z_1 \ldots, z_q] \qded M\})\hfill
  \tag{\textbf{A.1}}
\end{equation}

\item Else, if $(r\sigma)\downR$ is ground, then
\begin{equation}
\begin{disparray}{r@{}l}
  (\Phi, \Psi) \Longrightarrow (&\Phi \cup \{M_0 \ded (r\sigma)\downR\},\\
  &\Psi \cup \{\forall z_1, \ldots, z_q. D[M_1, \ldots, M_{n+p}, z_1
  \ldots, z_q] \qded M_0\})
 \end{disparray}
  \tag{\textbf{A.2}}
\end{equation}

where $M_0=D[M_1, \ldots, M_{n+p}, \cst a, \ldots, \cst a]$ for some fixed
public constant~$\cst a$.

\smallskip{}

\item Otherwise, % and if no rule \textbf{A.1}, \textbf{A.2} or \textbf{B.$j$}
%   is strictly applicable,
%  \hfill $(\Phi, \Psi) \Longrightarrow \bot$ \hfill\null
%\begin{equation}
$(\Phi, \Psi) \Longrightarrow \bot$
\hfill (\textbf{A.3})
\end{enumerate}
\end{minipage}

\bigskip

\textbf{B. Inferring deduction facts and equations syntactically}

\smallskip{}

Assume that $M_0 \ded t_0, \ldots, M_{n} \ded t_{n} \in \Phi \qquad
t = f(t_1, \ldots, t_n) \in \st(t_0) \qquad f\in\Fpub$

\medskip{}

\begin{minipage}{\textwidth}
\begin{enumerate}%[(1)]
\item If there exists $M$ such that $(M \ded t) \in \Phi$, 
\begin{equation}
(\Phi, \Psi) \Longrightarrow (\Phi, \Psi \cup \{f(M_1, \ldots, M_n)
  \bowtie M\})
  \tag{\textbf{B.1}}
\end{equation}

\item Otherwise, 
$(\Phi, \Psi) \Longrightarrow (\Phi \cup \{f(M_1,
  \ldots, M_n) \ded t\}, \Psi)$ \hfill (\textbf{B.2})
\end{enumerate}
\end{minipage}

\end{flushleft}

\caption{Transformation rules}
\label{fig:rules}
\end{figure}

The main part of our procedure consists in saturating a state $(\Phi,
\Psi)$ by means of the transformation rules described in
Figure~\ref{fig:rules}.  The \textbf{A} rules are designed for
applying a rewrite step on top of existing deduction facts. If the
resulting term $(r\sigma)\downR$ is already deducible (in some specific sense that we make precise below) then a corresponding equation is added
(rule~\textbf{A.1}); or else if it is ground, the corresponding
deduction fact is added to the state (rule~\textbf{A.2}); otherwise,
the procedure may fail (rule~\textbf{A.3}). The \textbf{B} rules are
meant to add syntactically deducible subterms (rule~\textbf{B.2}) or
related equations (rule~\textbf{B.1}).

\medskip{}

%\textcolor{blue}{ %%Mathieu:
%Actually, to avoid failure, we can apply~\textbf{A.1} in more cases (it is not
%needed that the resulting term is syntactically deducible). 
%For this reason,
For technical reasons, rule~\textbf{A.1} is parametrized by a
function~$\Ctx$ that outputs either a recipe $M$ or the special symbol $\bot$.
%with values of the form $M$ or~$\bot$,
%in $\Fpub[\W] \cup \{\bot\}$,
This function has to satisfy the following properties:
\begin{enumerate}[(a)]
\item\label{item:ctx_sufficient} if $\phi \vdash t\downR$, then
  $\Ctx(\phi \vdash^?_\R t) \neq \bot$;
\item\label{item:ctx_correct} if $M=\Ctx(\phi \vdash^?_\R t)$ then there exists~$s$ such that $M \ded_\phi s$ and $t \rR^* s$.
(This justifies the notation $\phi \vdash^?_\R t$ used to denote a specific deducibility problem.)
\end{enumerate}
Property~(\ref{item:ctx_sufficient}) ensures that the rules
transform a state into a state (and more precisely that the resulting frame in \textbf{(A.2)}
is still one-to-one). 
Property~(\ref{item:ctx_correct}) guarantees the soundness of the new equation in \textbf{(A.1)}. Requiring $t \rR^* s$ instead $t =_\E s$ is necessary for the proof of completeness.
In what follows, a \emph{function $\Ctx$} is any function satisfying the two properties (\ref{item:ctx_sufficient}) and (\ref{item:ctx_correct}).
%that
%when the function $\Ctx$ outputs a term $M$ then~$M$ is a recipe of~$s$ that
%is equal to~$t$ modulo~$\E$. This property is needed to ensure soundness of our procedure. 
%For completeness, we have to ensure that~$s$ is not further to its $\R$-reduced form
%than~$t$.
%}  

\medskip{}

A simple choice for $\Ctx(\phi \vdash^?_\R t)$
is to solve the deducibility problem $\phi \vdash^? t\downR$ in the
empty equational theory, and then return a corresponding recipe $M$,
if any. (This problem is easily solved by induction on~$t\downR$.)
We will see in Section~\ref{sec:fail} that this choice is 
sufficient to avoid failure for a large class of equational theories, namely
the class of layered convergent theories. However the proof of this fact relies on an intermediate result that uses a different choice of $\Ctx$.
%
%Yet, optimizing the function $\Ctx$ is a nontrivial task: on the
%one hand, letting $\Ctx(\phi \vdash^?_\R t, \Psi,\alpha) \neq \bot$
%for more values $\phi$, $t$, $\Psi$, $\alpha$ makes the procedure more
%likely to succeed (see Example~\ref{ex:ctx}); on the other hand, it is computationally more
%demanding. We explain in Section~\ref{sec:fail} the choice of~$\Ctx$ made in our implementation.

\begin{example}
\label{ex:saturationEnc}
  Consider the frame $\varphi_0$ previously described in
  Example~\ref{ex:static}. We can apply rule~\textbf{A.1} as
  follows. Consider the rewrite rule $\dec(\enc(x,y), y) \to x$, the
  decomposition $D_2$ given in Example~\ref{ex:decompo} and $t_1 =
  \enc(\cst c_0, \cst k)$. 
 We have that 
$\Init(\varphi_0) = (\varphi_0, \emptyset)
  \Longrightarrow (\varphi_0, \{\dec(\w_1,\w_2) \bowtie \cst c_0\}).$ 
In
  other words, since we know the key~$\cst k$ through~$\w_2$, we can check
  that the decryption of~$\w_1$ by~$\w_2$ leads to the public
  constant~$\cst c_0$. 
Next we apply rule \textbf{B.1} as follows:
\[
 (\varphi_0, \{\dec(\w_1,\w_2) \bowtie
  \cst c_0\})\Longrightarrow(\varphi_0, \{\dec(\w_1,\w_2) \bowtie \cst c_0,
  \enc(\cst c_0,\w_2) \bowtie \w_1\}).
\]
No more rules can then modify the
  state.
Similarly for $\varphi_1$, we obtain that:
\[
\begin{array}{rcl}
\Init(\varphi_1) &=& (\varphi_1, \emptyset)\\
&  \Longrightarrow &(\varphi_1, \{\dec(\w_1,\w_2) \bowtie \cst c_1\}) \\
&\Longrightarrow & (\varphi_1, \{\dec(\w_1,\w_2) \bowtie \cst c_1,
  \enc(\cst c_1,\w_2) \bowtie \w_1\}).
\end{array}
\]
\end{example}

\begin{example}
\label{ex:saturationHom}
Consider the frame $\varphi$ described in
Example~\ref{ex:quantifiedeq}. We can apply rule~\textbf{A.1} as follows.
Consider the rewrite rule $\dec(\langle x, y \rangle, z) \to \langle
\dec(x,z), \dec(y,z) \rangle$, the decomposition $D_4$ given in
Example~\ref{ex:decompo} and $t_1 = \langle \enc(\cst c_0, \cst k), \enc(\cst
c_1, \cst k) \rangle$. We have that $r\sigma =  \langle \dec(\enc(\cst
c_0, \cst k), z_1), \dec(\enc(\cst c_1, \cst k), z_1) \rangle$, and thus
$\Init(\varphi)  \Longrightarrow \bot$.
We have that $r\sigma\downR = r\sigma$. The condition required in case~{(1)} is not fulfilled
and the condition stated in case (2) is false.

However, note that another strategy of rules application allows us to consider
this decomposition. For this, it is sufficient to apply first \textbf{B} rules to
 add the deduction facts $\proj_1(\w_1) \ded \enc(\cst c_0, \cst k)$ and
$\proj_2(\w_1) \ded \enc(\cst c_1, \cst k)$. Now, we have that $r\sigma\downR$ is
syntactically deducible: the condition required in case~{(1)} is full-filled
and we finally add the equation:
$\forall z_1.  \dec(\w_1,z_1) \bowtie \langle \dec(\proj_1(\w_1), z_1),
\dec(\proj_2(\w_1), z_1) \rangle$.
\end{example}

% \begin{example}
% \label{ex:ctx}
% \commentaire{S.D.: Je voulais trouver un ex pour faire apparaitre 
% que le choix simple de Ctx pourrait mener a failure.
% }

% \commentaire{V.C.: Un exemple qui montre que ca peut servir a avoir moins de failure mais il en reste
% }

% $f,g,h$ public symbols. We consider the following convergent (I think) rewrite system:
% \begin{eqnarray*}
%   h(g(f(x),y),z) & \rightarrow & h(g(x,f(y)),z) \\
% h(g(f(x),f(y)),z) & \rightarrow & h(g(x,y),z) \\
% h(g(x,f^2(y)),z) & \rightarrow & h(g(x,y),z)
% \end{eqnarray*}
% We consider the frame
% \[\{\w_1\ded g(f^2(a),b);\w_2\ded g(a,f^2(b));
% \w_3\ded f(a);
% \w_4\ded f(b)\}
% \]
% Then $h(\w_1,z)\rightarrow h(g(f(a),f(b)),z)$, syntactically deducible while its normal form
% $h(g(a,b),z)$ is not syntactically deducible.
% However, we would still get a failure with $\w_3$ and $\w_4$.
% Indeed $h(g(\w_3,y),z)\rightarrow h(g(a,f(y)),z)$ which is not syntactically deducible.
% \end{example}

We write $\Longrightarrow^*$ for the transitive and reflexive closure of
$\Longrightarrow$.
The definitions of $\Ctx$  and of the transformation rules
ensure that whenever $S\Longrightarrow^* S'$ and $S$ is a state, then $S'$
is also a state, with the same parameters unless~${S'=\bot}$.
%Mathieu: déjà dit en haut
%\textcolor{blue}{In particular, property~(\ref{item:ctx_sufficient}) of $\Ctx$ ensures that the
%frame remains one-to-one along a derivation.}

\subsection{Main theorem} 
We now state the soundness and the completeness of the transformation
rules provided that a \emph{saturated state} is reached, that is, a
state $S\neq \bot$ such that $S \Longrightarrow S'$ implies $S' = S$.
The technical lemmas involved in the
proof of this theorem
are detailed in Section~\ref{sec:completeness}.

\medskip{}

\newsavebox{\theocompleteness}
\sbox{\theocompleteness}{\vbox{%
\begin{theorem}[soundness and completeness]\label{theo:soundcomp}
Let~$\E$ be an equational theory generated by a convergent rewrite
system~$\R$.  Let $\varphi$ be an initial frame and $(\Phi,\Psi)$ be a
saturated state such that $\Init(\varphi) \Longrightarrow^*
(\Phi,\Psi)$. %We have that
\begin{enumerate}
\item For all $M \in \Fpub[\param(\varphi)]$ and $t \in \F[\emptyset]$, we
  have that:
\[
M \varphi =_\E t \quad\Leftrightarrow\quad\exists N \text{ such that }
      \Psi \models M \qded N \text{ and } N \ded_{\Phi} t\downR.
\]

\item  For all $M$, $N \in \Fpub[\param(\varphi) \cup \X]$, we have that:
\[
    M \varphi =_\E N \varphi \,\Leftrightarrow\,
    \Psi \models M \qded N.
\]

\end{enumerate}
\end{theorem}
}}

\noindent\usebox{\theocompleteness}

We note that this theorem applies to any saturated state
reachable from the initial frame. Moreover, 
while the saturation procedure is sound and complete, it may not
terminate, or it may \emph{fail} if rule~\textbf{A.3} becomes the only applicable
rule at some point of computation.
In Section~\ref{sec:fail} and Section~\ref{sec:termination}, we explore several sufficient conditions
to prevent failure and ensure termination.

\subsection{Application to deduction and static equivalence}

Decision procedures for deduction and static equivalence modulo~$\E$
follow from Theorem~\ref{theo:soundcomp}.

\smallskip{}

\paragraph*{Algorithm for deduction}
Let~$\varphi$ be an initial frame and~$t$ be a ground term.
%  Is~$t$ deducible
% from~$\varphi$ modulo~$\E$, i.e. does there exists~$M$ such that ${M
% \ded_{\varphi}^\E t}$?
The procedure for checking $\varphi \vdash_\E t$ runs as follows:
\begin{enumerate}
\item Apply the transformation rules %\fbox{as prescribed} 
to obtain (if any) a
  saturated state $(\Phi, \Psi)$ such that $\Init(\varphi)
  \Longrightarrow^* (\Phi, \Psi)$;
\item Return \emph{yes} if there exists $N$ such that $N \ded_\Phi 
  t\downR$ (that is, the $\R$-reduced form of~$t$ is syntactically deducible from~$\Phi$); otherwise return~\emph{no}.
\end{enumerate}

\begin{proof}
If the algorithm returns \emph{yes}, this means that
 there exists $N$
  such that $N \ded_{\Phi}  t\downR$. Thanks to Theorem~\ref{theo:soundcomp} (1),
  we have that $N\varphi =_\E t$, i.e. $N \ded_{\varphi}^\E t$.

Conversely, if $t$ is deducible from $\varphi$, then  there exists $M$ such that
  $M\varphi =_\E t$. By Theorem~\ref{theo:soundcomp} (1), there exists $N$ such that $N \ded_{\Phi}
  t\downR$. The algorithm returns \emph{yes}. \qed
\end{proof}

\begin{example}
Consider the frame $\varphi_0 =  \{\w_1 \ded \enc(\cst c_0, \cst k),
\w_2 \ded \cst k\}$ introduced in Example~\ref{ex:deducibility} and let $t_1 =
\langle \cst k, \cst k \rangle$ and $t_2 = \cst c_0$.
Let $(\Phi_0, \Psi_0)$ be the saturated state described in
Example~\ref{ex:saturationEnc}. We have that:
\[
(\Phi_0, \Psi_0) = (\varphi_0, \{\dec(\w_1,\w_2) \bowtie \cst c_0,
  \enc(\cst c_0,\w_2) \bowtie \w_1\}).
\]
Then, it is easy to see that our algorithm for deduction will return
\emph{yes} for both terms~$t_1$
and~$t_2$. Indeed, those terms are syntactically deducible from~$\varphi_0$.
\end{example}

\paragraph*{Algorithm for static equivalence}

Let $\varphi_1$ and $\varphi_2$ be two initial frames. The
procedure for checking $\varphi_1 \approx_\E \varphi_2$ runs as follows:

\begin{enumerate}
\item Apply the transformation rules to obtain (if possible)
  two saturated states $(\Phi_1,\Psi_1)$ and $(\Phi_2, \Psi_2)$ such
  that $\Init(\varphi_i) \Longrightarrow^* (\Phi_i,\Psi_i)$,
  ${i=1,2}$;
%  and $(\varphi_2\downR, \emptyset) \Longrightarrow^* (\Phi_2,\Psi_2)$;
%  (if such states do not exist the procedure fails;)
\item For $\{i,j\}=\{1,2\}$, for every equation $(\forall z_1,\ldots,z_\ell. M
  \qded N)$ in $\Psi_i$, check that $M\varphi_j =_\E N\varphi_j$~---~that is, in
  other words, $(M\varphi_j) \downR = (N\varphi_j) \downR$;
\item If so return \emph{yes}; otherwise return \emph{no}.
\end{enumerate}

\begin{proof}
If the algorithm returns \emph{yes}, this means that $M\varphi_2 =_{\E}
N\varphi_2$ for every
  equation $(\forall z_1, \ldots, z_\ell. M \qded N)$ in $\Psi_1$.
  Let $M \qded N \in \eq_\E(\varphi_1)$. By definition of
  $\eq_\E(\varphi_1)$, we have that $M\varphi_1 =_\E N\varphi_1$. Thanks
  to Theorem~\ref{theo:soundcomp} (2), we have that $\Psi_1 \models M
  \qded N$. As all the equations in $\Psi_1$ are satisfied by
  $\varphi_2$ modulo $\E$, we deduce that $M\varphi_2 =_\E N\varphi_2$,
  i.e. $M \qded N \in \eq(\varphi_2)$.
  The other inclusion, $\eq_\E(\varphi_2) \subseteq
  \eq_{\E}(\varphi_1)$, is proved in the same way.

\smallskip{}

Conversely, assume now that $\varphi_1 \approx_\E \varphi_2$,
  i.e. $\eq_\E(\varphi_1) = \eq_\E(\varphi_2)$. Consider a quantified equation $\forall
  z_1,\ldots,z_\ell. M \qded N$ in $\Psi_1$ and let us show that
  $M\varphi_2 =_\E N\varphi_2$. (The other case is done in a similar
  way, and we will conclude that the algorithm returns \emph{yes}.)
  Let $\cst c_1, \ldots, \cst c_\ell$ be free public constants not occurring in
  $M$ and $N$, and let $(M',N')=(M,N)\{z_1 \mapsto \cst c_1,\ldots,z_\ell
  \mapsto \cst c_\ell \}$. Since $\Psi_1 \models M' \qded N'$, by
  Theorem~\ref{theo:soundcomp} (2), we have that $M' \varphi_1 =_\E N'
  \varphi_1$. Besides, $M'$ and $N'$ are ground and $\param(M',N')
  \subseteq \param(\Psi_1) \subseteq \param(\varphi_1)$. 
Thus, $(M' \qded N') \in
  \eq_\E(\varphi_1) \subseteq \eq_\E(\varphi_2)$ and $M' \varphi_{2}
  =_{\E} N'\varphi_{2}$. As the constants $\cst c_1, \ldots, \cst c_\ell$ are
  free in $\E$ and do not occur in $M$ and $N$, by replacement, we
  obtain that $M \varphi_{2} =_{\E} N\varphi_{2}$. \qed
\end{proof}

\begin{example}
Consider the frames $\varphi_i =  \{\w_1 \ded \enc(\cst c_i, \cst k),
\w_2 \ded \cst k\}$ introduced in Example~\ref{ex:static}.
Let $(\Phi_0, \Psi_0)$ and $(\Phi_1, \Psi_1)$ be the two saturated states  
described in Example~\ref{ex:saturationEnc}. We have
that $\dec(\w_1,\w_2) \bowtie \cst c_0 \,\in \, \Psi_0$, and
\[
(\dec(\w_1,\w_2)\varphi_1 =_{\E_\enc} \cst c_1 \not = _{\E_\enc} \cst c_0 =
\cst c_0\varphi_1.
\] Hence, our algorithm returns \emph{no}. The two frames
$\varphi_0$ and $\varphi_1$ are not statically equivalent.
\end{example}

%%% Local Variables: 
%%% mode: latex
%%% TeX-master: "main"
%%% End: 

\section{Soundness and completeness of the
saturation}\label{sec:completeness}

%The proof of Theorem~\ref{theo:soundcomp} is based on three main lemmas.

%
%We start with a few technical lemmas and finally prove the main theorem.
%
%% Let $(\varphi\downR, \emptyset)$ be our initial state. 
%% Along the saturation process, we only add deduction facts and equations to the
%% state. Actually Lemma~\ref{lem:sound} states that:
%% \begin{itemize}
%% \item deducible terms from $\Phi$ in the empty equational theory are also
%%   deducible from our intial frame modulo~$\E$;
%% \item each equation that is a  consequence of~$\Psi$ is a visible equation of
%%   our initial frame modulo $\E$.
%% \end{itemize}
%

The goal of this section is to prove Theorem~\ref{theo:soundcomp}. Section~\ref{subsec:soundness}
is devoted to establish soundness of our saturation procedure, i.e. the $\Leftarrow$
direction of Theorem~\ref{theo:soundcomp}. Showing the other direction, i.e.
completeness, is more
involved and is detailed in Section~\ref{subsec:completeness}.

\subsection{Soundness}
\label{subsec:soundness}
First, the transformation rules are sound in the sense that, along the
saturation process, we add only deducible terms and valid equations
with respect to the initial frame.

\begin{lemma}[soundness]
\label{lem:sound}
Let~$\varphi$ be an initial frame and~$(\Phi, \Psi)$ be a state such
that $\Init(\varphi) \Longrightarrow^* (\Phi,
\Psi)$. Then, we have that
\begin{enumerate}
\item \label{soundcond1} $M \ded_{\Phi} t \; \Rightarrow \; M
  \varphi =_\E t$ \;\; for all $M\in\Fpub[\dom(\varphi)]$ and
  $t\in\F[\emptyset]$;
\item \label{soundcond2} $\Psi \models M \bowtie N \; \Rightarrow \; 
M\varphi =_\E N\varphi$ \;\; for all $M,N\in\Fpub[\dom(\varphi) \cup \X]$.
\end{enumerate}
\end{lemma}

\begin{proof}
We prove this result by induction on the derivation
$\Init(\varphi) \Longrightarrow^* (\Phi, \Psi)$.

\medskip{}

\noindent\emph{Base case:}
We have that $(\Phi, \Psi) = \Init(\varphi)$
 and we easily conclude. 

\medskip{}

\noindent\emph{Induction case:} 
 In such a case, we have
$\Init(\varphi) \Longrightarrow^* (\Phi', \Psi') \Longrightarrow (\Phi,
\Psi)$.

Let us first notice two facts.
\begin{enumerate}
\item Let~$M$ and~$t$ be such that $M \ded_{\Phi} t$. 
By definition of~$\ded_{\Phi}$, %this means that
there exist a public context~$C$ and some deduction facts $M'_1 \ded t'_1, \ldots, M'_n \ded t'_n \in \Phi$ such that
$M = C[M'_1, \ldots M'_n]$ and $t = C[t'_1, \ldots, t'_n]$. In order to
prove~\emph{\ref{soundcond1}.}, it is sufficient to show that $M' \ded_{\varphi}^\E t'$
for every $M' \ded t' \in \Phi$. By induction hypothesis, this holds for
the deduction facts in $\Phi'$, thus it remains to show that $M' \ded_{\varphi}^\E t'$ for
every fact $M' \ded t' \in \Phi - \Phi'$.

\item Let $M, N$ be two terms such that $\Psi \models M \bowtie N$.
  To establish \emph{\ref{soundcond2}.}, it is sufficient to prove
  that $M' \varphi =_\E N' \varphi$ for every $(\forall z_1,\ldots, z_q. M'
  \bowtie N')$ in $\Psi$. By induction hypothesis, this holds for the
  equations in $\Psi'$, thus it remains to show that $M'\varphi =_\E
  N'\varphi$ for every equation $(\forall z_1,\ldots, z_q. M' \bowtie N')$
  in $\Psi - \Psi'$.
\end{enumerate}

\noindent Next we perform a case analysis on the inference rule used
in $(\Phi', \Psi') \Longrightarrow (\Phi, \Psi)$.

\medskip{}

First, consider % that the rule involved in this step is
the case of rule $\textbf{A}$.
Let $l \to r \in \R$ be the rewrite rule, $D$ the decomposition, and $M_1
\ded t_1, \ldots, M_{n+p} \ded t_{n+p}$ the facts involved in this
step.

\smallskip{}

\noindent \emph{Rule \textbf{A.2}}:
%In order to conclude for this case, 
We need to show that
\begin{itemize}
\item $D[M_1, \ldots ,M_{n+p}, \cst a, \ldots, \cst a] \varphi =_\E (r\sigma)\downR$, and
\item $D[M_1, \ldots, M_{n+p}, z_1, \ldots, z_q] \varphi =_\E D[M_1,
  \ldots, M_{n+p}, \cst a, \ldots, \cst a]\varphi$.
\end{itemize}
We note that $D[t_1, \ldots, %t_n,t_{n+1}, \ldots,
t_{n+p}, z_1, \ldots, z_q]=l\sigma \to r\sigma \to^* (r\sigma)\downR$.
Besides, by induction hypothesis we have that $M_i \varphi =_\E t_i$ for $1
\leq i \leq n+p$.
Given that~$(r\sigma)\downR$ is ground, and applying the substitution
$\{z_1 \mapsto \cst a, \ldots, z_q \mapsto \cst a\}$ to the equation $D[t_1,
\ldots, t_{n+p}, z_1, \ldots, z_q] =_\E (r\sigma)\downR$,
we obtain:
\begin{eqnarray*}
D[M_1, \ldots, M_{n+p}, z_1, \ldots, z_q] \varphi &\; =_\E \; & D[t_1,
\ldots, %t_n,t_{n+1}, \ldots,
t_{n+p}, z_1, \ldots, z_q] \\
&=_\E& (r\sigma)\downR\\
&=_\E& D[t_1,
\ldots, %t_n,t_{n+1}, \ldots,
t_{n+p}, \cst a, \ldots, \cst a]\\
&=_\E& D[M_1, \ldots ,M_{n+p}, \cst a, \ldots, \cst a]
\varphi
\end{eqnarray*}

\noindent \emph{Rule \textbf{A.1}}: We need to show $D[M_1, \ldots,
M_{n+p}, z_1, \ldots, z_q] \varphi =_\E M \varphi$.
As before, we have $D[M_1, \ldots, M_{n+p}, z_1, \ldots,
z_q]\varphi =_\E (r\sigma)\downR$.
%
%Besides, ${M\varphi =_\E (r\sigma\downR)}$ due to the fact that $M
%\ded_{\Phi \cup \{z_1 \ded z_1, \ldots, z_q \ded z_q\}}
%(r\sigma)\downR$ and using the induction hypothesis.
We also know that there exists~$s$ such that $M\ded_{\Phi^+} s$ and $r\sigma
\to^*_{\R} s$ where $\Phi^+ = \Phi \cup
\{z_1 \ded z_1,\ldots, z_q \ded z_q\}$
%Mathieu: Ici et ailleurs, c'est \phi qu'il faudrait utiliser à la place de \Phi^+
thanks to
property~(\ref{item:ctx_correct}) of $\Ctx$. 
Let $\theta$ be the substitution
$\{z_1 \mapsto \cst a, \ldots, z_q \mapsto \cst a\}$.
We have that $M\theta\ded_{\Phi} s$.
Hence, using the induction hypothesis,
we have that {$M\theta\varphi =_{\E}
s$} thus {$M\varphi =_{\E}
s$}, i.e. $M\varphi =_{\E} (r\sigma)\downR$. This allows us to conclude.
\smallskip{}

\noindent \emph{Rule \textbf{A.3}}:
In such a case, the result trivially holds.

\medskip{}

\noindent Second, we consider the case of \textbf{B} rules. Let $t = f(t_1,\ldots,
t_n) \in \st(t_0)$, $f \in \F_{\pub}$ and $M_0 \ded t_0, \ldots, M_{n} \ded
t_{n} \in \Phi$ be involved in the step
$(\Phi', \Psi') \Longrightarrow (\Phi, \Psi)$.

\medskip{}

\noindent \emph{Rule \textbf{B.1}}:
%In such a case, it remains  to
By induction hypothesis, $M_i \varphi =_\E t_i$ for every $1\leq i
\leq n$, hence $f(M_1,\ldots, M_n) \varphi =_\E f(t_1,\ldots, t_n)
= t$.\medskip{}

\noindent \emph{Rule \textbf{B.2}}: By induction hypothesis, $M_i
\varphi =_\E t_i$ for every $1\leq i \leq n$ and $M \varphi =_\E t$,
hence $f(M_1,\ldots, M_n) \varphi =_\E f(t_1,\ldots, t_n) = t =_\E M\varphi$.
% In this case, we need to show that $f(M_1,\ldots, M_n)\varphi =_E M \varphi$
% where $M \ded t \in \Phi$ and $t = f(t_1,\ldots, t_n)$.
% By induction hypothesis, we have that $M\varphi =_E t$ and  $M_i\varphi=_E
% t_i$ for $1 \leq i \leq n$. Hence, we easily
% deduce that $f(M_1,\ldots, M_n)\varphi =_E M\varphi$.
%
%\hfill$\Box$
\qed
\end{proof}

%%%%%%%%%%%%%%%%%%%%%%%%%%%%%%%%%%%%%%%%%%%%%%%%%%%%%%%%%%%%%%%

\subsection{Completeness}
\label{subsec:completeness}

The next three lemmas are dedicated to the completeness of \textbf{B} rules
(Lemma~\ref{lem:compbasebase} and Lemma~\ref{lem:compbase}) and
\textbf{A} rules (Lemma~\ref{lem:comp}). 

\medskip{}

Lemma~\ref{lem:compbasebase} ensures that a 
saturated state $(\Phi, \Psi)$ contains all the deduction facts~${M \ded t}$ where~$t$ is a
subterm of~$\Phi$ that is syntactically deducible, whereas
Lemma~\ref{lem:compbase} ensures that
saturated states account for all the syntactic equations possibly
visible on the frame.

\newsavebox{\lemcompsyntded}
\sbox{\lemcompsyntded}{\vbox{%
\begin{lemma}[completeness, syntactic deduction]
\label{lem:compbasebase}
Let $(\Phi, \Psi)$ be a state, $M_0 \ded t_0 \in \Phi$. Let $N$, $t$
be two terms such that $t \in \st(t_0)$ and $N \ded_{\Phi} t$. Then
there exists $(\Phi', \Psi')$ and $N'$ such that:
\begin{itemize}
\item $(\Phi, \Psi) \Longrightarrow^* (\Phi', \Psi')$ using \textbf{B} rules, and
\item $N'\ded t \in \Phi'$ and $\Psi'\models N \bowtie N'$.
\end{itemize}
\end{lemma}
}}

\noindent\usebox{\lemcompsyntded}
The proof of Lemma~\ref{lem:compbasebase} is postponed to the appendix. It uses
 a simple
induction on the context~$C$ witnessing the fact that~$t$ is syntactically
deducible from~$\Phi$.

\newsavebox{\lemcompsynteq}
\sbox{\lemcompsynteq}{\vbox{%
\begin{lemma}[completeness, syntactic equations]
\label{lem:compbase}
Let $(\Phi, \Psi)$ be a state, and $M$, $N$ be two terms such that
${M\ded_{\Phi} t}$ and~${N \ded_{\Phi} t}$ for some term~$t$. 
Then there exists  $(\Phi', \Psi')$ such that:
\begin{itemize}
\item  $(\Phi, \Psi) \Longrightarrow^* (\Phi', \Psi')$ using \textbf{B} rules, and
\item  $\Psi'\models M \bowtie N$.
\end{itemize}
\end{lemma}
}}

\noindent\usebox{\lemcompsynteq}

\begin{proof}(sketch)
Let $C$, $C'$ be the contexts witnessing ${M\ded_{\Phi} t}$
  and~${N \ded_{\Phi} t}$. Assume that $C$ is smaller than $C'$.  The proof is
  done by induction on $C$. When $C$ is reduced to an hole, we apply
  Lemma~\ref{lem:compbasebase} to conclude. Otherwise, we have that $C = f(C_1,\ldots, C_r)$
  and $C' = f(C'_1,\ldots, C'_r)$. We easily conclude by
  applying our induction hypothesis on $C_i, C'_i$ for each $1\leq i \leq r$. 
The detailed proof is presented in appendix~\ref{app:app-comp}. \qed
\end{proof}

%%%%%%%%%%%%%%%%%%%%%%%%%%%%%%%%%%%%%%%%%%%%%%%%%%%%%%%%%%%%%%%%%%%%%%%%%%
%                                                                        %
%                    COMPLETENESS - CONTEXT REDUCTION                    %
%                                                                        %
%%%%%%%%%%%%%%%%%%%%%%%%%%%%%%%%%%%%%%%%%%%%%%%%%%%%%%%%%%%%%%%%%%%%%%%%%%

Now, we know that terms that are syntactically deducible from the frame 
and syntactic equation visible on the frame will be added during our
saturation procedure. It remains to take into account the underlying equational
theory. This is the purpose of Lemma~\ref{lem:comp} that 
deals with the reduction of a deducible term
along the rewrite system~$\R$.  Using that $\R$ is convergent, this
allows us to prove that every deducible term % modulo $\E$
from a saturated frame is syntactically deducible.

\newsavebox{\lemcompred}
\sbox{\lemcompred}{\vbox{%
\begin{lemma}[completeness, context reduction]
\label{lem:comp}
Let $(\Phi, \Psi)$ be a state and $M$, $t$, $t'$ be three terms such that $M
\ded_{\Phi} t$ and $t \rR t'$. Then, either $(\Phi, \Psi)
\Longrightarrow^* \bot$ 
or there exist $(\Phi', \Psi')$, $M'$ and $t''$ such that
\begin{itemize}
\item 
$(\Phi, \Psi) \Longrightarrow^* (\Phi', \Psi')$,
\item 
$M'\ded_{\Phi'} t''$ with $t'\rR^* t''$, and
\item 
$\Psi'\models M \bowtie M'$.
\end{itemize}

Besides, in both cases, the corresponding derivation from $(\Phi,
\Psi)$ can be chosen to consist of a number of \textbf{B} rules,
possibly followed by one instance of \textbf{A} rule involving the
same rewrite rule $l \to r$ as the rewrite step $t \rR t'$.
\end{lemma}
}}

\noindent\usebox{\lemcompred}
 
\begin{proof}(sketch)
The detailed proof of Lemma~\ref{lem:comp} is left to the appendix. We describe here its main arguments.
Since $t \rR t'$, there exist a
position~${\alpha}$, a substitution~$\sigma$ and a
  rewrite rule ${l \to r \in \R}$ such that $t|_{\alpha} = l\sigma$ and
  $t'= t[r\sigma]_{\alpha}$. 
Let~$C$ be a context witnessing the fact that  $M \ded_{\Phi} t$.
Since terms in $\im(\Phi)$ are $\R$-reduced, $\alpha$ is actually a position in $C$.
Thus, the rewriting step mentioned above corresponds to a proper
  $(n,p,q)$-decomposition $D$ of $l$:
    $l = D[l_1,\ldots,l_n,y_1,\ldots y_{p}, z_1,\ldots z_{q}]$. 
%   $M_1 \ded t_1$, \ldots,  $M_n \ded t_n$ in $\Phi$,
%  \item $N_1$, \ldots, $N_{p+q}$
%  such that
%  \begin{itemize}
%  \item for every $1 \leq i \leq n$, $t_i = l_i \sigma$,
%  \item for every $1 \leq j \leq p$, $N_j \ded_\Phi y_j \sigma$, and
%  \item for every $1 \leq k \leq q$, $N_{p+k} \ded_\Phi z_k \sigma$.
%  \end{itemize}
We can show that
    $M|_{\alpha} \ded_{\Phi} l\sigma$ and 
    $D[M_1,\ldots, M_n, N_1, \ldots, N_{p+q}] \ded_{\Phi} l\sigma$
where    
\begin{itemize}
\item $M_1 \ded t_1$, \ldots,  $M_n \ded t_n$ are deduction facts in $\Phi$,
\item   for every $1 \leq j \leq p$, $N_j \ded_\Phi y_j \sigma$, and
\item for every $1 \leq k \leq q$, $N_{p+k} \ded_\Phi z_k \sigma$.
  \end{itemize}

  Thus, by Lemma~\ref{lem:compbase}, there exists a derivation $(\Phi,
  \Psi) \Longrightarrow^* (\Phi_1, \Psi_1)$ using~\textbf{B}~rules such
  that $\Psi_1 \models M|_{\alpha} \qded D[M_1,\ldots, M_n, N_1, \ldots, N_{p+q}]$.

Besides, $y_j\sigma$ is a subterm of some
  $l_{i}\sigma=t_{i}$. Since $N_j \ded_\Phi y_j \sigma$, by
  applying Lemma~\ref{lem:compbasebase} repeatedly, we deduce that
  there exist some term $M_{n+1}$, \ldots, $M_{n+p}$ and a derivation
  $(\Phi_1, \Psi_1) \Longrightarrow^* (\Phi_2, \Psi_2)$ using
  \textbf{B}~rules such that for all $j$,
  \begin{itemize}
  \item $M_{n+j} \ded y_j \sigma$ is in $\Phi_2$, and
  \item $\Psi_2 \models M_{n+j} \qded N_j$.
  \end{itemize}

  Let $N = D[M_1,\ldots, M_{n+p}, N_{p+1}, \ldots, N_{p+q}]$.  We
  deduce that $N \ded_{\Phi_2} l\sigma$, and \[\Psi_2 \models M|_{\alpha}
  \qded D[M_1,\ldots, M_n, N_1, \ldots, N_{p+q}] \qded N\]

%  \smallskip{}

  \noindent We now consider the application to $(\Phi_2,\Psi_2)$ of a
  \textbf{A}~rule that involves the rewrite rule $l\to r$, the
  decomposition $D$, the plain terms
  $(t_1,\ldots,t_{n+p})=(l_1,\ldots,l_n, y_1,\ldots,y_p)\sigma$. Depending on
  whether $(r\sigma)\downR$ is ground and $\Ctx(\Phi_2^+  \vdash^?_\R
  r\sigma') = \bot$, we conclude by applying \textbf{A.1},
  \textbf{A.2} or \textbf{A.3}. \qed
\end{proof}

\subsection{Main theorem}

We are now able to prove soundness and completeness of our transformation
rules provided that a saturated state is reached.

\noindent\usebox{\theocompleteness}

\begin{proof}
Let~$\varphi$ be an initial frame and~${(\Phi, \Psi)}$ be a saturated
state such that $\Init(\varphi) \Rightarrow^* (\Phi,\Psi)$.

\medskip{}

\noindent $1. (\Leftarrow)$
 Let $M$, $N$ and $t$ be such that $\Psi \models M \bowtie N$ and $N \ded_\Phi
 t\downR$ (thus in particular $N \ded^\E_\Phi
 t$).  Thanks to Lemma~\ref{lem:sound}, we have that $M\varphi =_\E N\varphi  =_\E t$.

\smallskip{}
\noindent \phantom{$1$} ($\Rightarrow$)
Let $M$ and $t$ be such that $M\varphi =_\E t$. We have that $M \ded_{\Phi} t_0
 \to^* t\downR$ for some term~$t_0$. We show the result by induction on $t_0$
equipped with the order $<$ induced by the rewriting relation ($t < t'$ if and 
only if $t'\to^+ t$).

\smallskip{}

\noindent \emph{Base case: $M \ded_{\Phi} t_0 = t\downR$.}
Let $N = M$,  we have $\Psi \models M \bowtie N$ and $N \ded_{\Phi} t\downR$.

\smallskip{}

\noindent \emph{Induction case: $M \ded_{\Phi} t_0 \to^+ t\downR$.}
Let~$t'$ be such that $M \ded_{\Phi} t_0 \to t'\to^* t\downR$.  Thanks
to Lemma~\ref{lem:comp} and since $(\Phi, \Psi)$ is already
saturated\footnote{Note that rule \textbf{A.3} is never applicable on a 
  saturated state.},  we deduce that there exist~$N'$ and~$t''$ such
that $N' \ded_{\Phi} t''$, $t'\to^* t''$, and $\Psi\models M \bowtie
N'$.  We have that $N'\ded_{\Phi} t''\to^* t\downR$ and $t'' \leq t' <
t_0$. Thus, we can apply our induction hypothesis and we obtain that
there exists~$N$ such that $\Psi \models N'\bowtie N$ and
$N\ded_{\Phi} t\downR$.

\medskip{}

\noindent $2. (\Leftarrow)$
By Lemma~\ref{lem:sound},
$\Psi \models M  \bowtie N$ implies $M\varphi =_\E N\varphi$. 

\smallskip{}

\noindent \phantom{$2$} ($\Rightarrow$)
Let $M$ and $N$ such that $M\varphi =_\E N\varphi$. This means that there
exists $t$ such that $M\varphi =_{\E} t$ and $N\varphi =_{\E} t$.
By applying~$1$, we deduce that there exists $M'$, $N'$ such that:
 $\psi \models M \bowtie M'$, $M'\ded_{\Phi} t\downR$,  
$\psi \models N \bowtie N'$ and $N'\ded_{\Phi} t\downR$.
Thanks to Lemma~\ref{lem:compbase} and since $(\Phi, \Psi)$ is already
saturated, we easily deduce that $\Psi \models M' \bowtie N'$, and thus
$\Psi \models M \bowtie N$.
\qed 
\end{proof}

We proved that saturated
frames yield sound and complete characterizations of deducible terms and
visible equations of their initial frames. Yet, the saturation
procedure may still not terminate, or fail due to rule~\textbf{A.3}.
%

%%% Local Variables: 
%%% mode: latex
%%% TeX-master: "main"
%%% End: 

\section{Non-failure}
\label{sec:fail}

As shown by the following example (from~\cite{CDK-cade09}), our procedure may fail.

\begin{example}
\label{ex:malleableenc}
Consider the theory $\E_\mal$ given below:
\[
\E_\mal = \{\dec(\enc(x,y),y) = x,\;
  \mal(\enc(x,y),z) = \enc(z,y) \}.
\]
The $\mal$ function symbol allows one to arbitrarily change the
plaintext of an encryption. Such a malleable encryption is not realistic. It 
is only used for illustrative purpose. 

By orienting from left to right the equations, we obtain a convergent
rewrite system. Thus, $\E_\mal$ is a convergent equational theory.
Let $\varphi = \{\w_1 \ded \enc(\cst \cst s, \cst k)\}$
where $\cst s$ and $\cst k$ are
private constants. The only rule that is applicable is an instance of 
an \textbf{A} rule. Consider the rewrite rule $\mal(\enc(x,y),z) \to
\enc(z,y)$ and the only deduction fact in $\Init(\varphi) = (\varphi,\emptyset)$. We obtain
$r\sigma\downR = \enc(z, \cst k)$. This term is not ground and the condition
required in case (1) is not fulfilled. Thus, we have that 
$\Init(\varphi) \Longrightarrow \bot$. Note that, since no other rule is
applicable, there is no hope to find a strategy of rule applications 
to handle this case.
\end{example}

In this section, we identify a class of theories, called 
\emph{layered convergent} theories, (a syntactically defined class of theories) for which failure is guaranteed not to occur.

%propose a syntactic criterion on the
%rewrite system~$\R$ to guarantee that failure will never happen.
% and/or termination is
%ensured.

\subsection{Layered convergent theories}
We prove that the algorithm never fails for \emph{layered convergent} theories.
Layered convergent theories consist in a generalization of subterm theories, considering each decomposition of the rewrite rules of the theory.

% Our first criterion is syntactic and ensures that the algorithm never
% fails. It is enjoyed by a large class of equational theories, called
% \emph{layered convergent}.

\begin{definition}[layered rewrite system] \label{def:layered} A
  rewrite system $\R$, and by extension its equational theory $\E$, are
  \emph{layered} if there exists an ascending chain of subsets
  $\emptyset = \R_0 \subseteq \R_1 \subseteq \ldots \subseteq \R_{N+1}
  = \R$ $(N \geq 0)$, such that for every $0 \leq i \leq N$, for every
  rule $l \to r$ in $\R_{i+1}-\R_{i}$, for every %corresponding %
  $(n,p,q)$-decomposition $l=D[l_1, \ldots, l_n, y_1, \ldots, y_p,
  z_1, \ldots, z_q]$, %either
  one of the following two conditions holds: 
  \begin{enumerate}[(i)]
  \item $\var(r) \subseteq \var(l_1, \ldots, l_n)$;%, or
  \item there exist $C_0, C_1,\ldots,C_k$ and $s_1,\ldots,s_k$ such that
    \begin{itemize}
    \item $r=C_0[s_1, \ldots, s_k]$;
    \item for each $1 \leq i \leq k$, $C_i[l_1, \ldots, l_n, y_1,
      \ldots, y_p,z_1, \ldots, z_q]$ rewrites to $s_i$ in zero or one
      step of rewrite rule in head position along $\R_i$.
    \end{itemize}
  \end{enumerate}
  In the latter case, we say that the context $C=C_0[C_1,\ldots,C_k]$
  is \emph{associated} to the decomposition $D$ of $l \to r$. Note
  that $C[l_1, \ldots, l_n, y_1, \ldots, y_p,z_1, \ldots, z_q]
  \to_{\R_i}^* r$. 
\end{definition}

The large class of weakly subterm convergent is an (easy) particular case of layered convergent theories.
\begin{lemma}
  Any weakly subterm convergent rewrite system~$\R$ is layered
  convergent.
\end{lemma}

\begin{proof}
Let $N=0$ and $\R_1=\R$.  For any $l \to r$ in
  $\R$ and for every decomposition $l=D[l_1, \ldots, l_n, y_1, \ldots,
  y_p, z_1, \ldots, z_q]$, the term $r$ is a subterm of $l$, thus either $r =
  C[l_1, \ldots, l_n, y_1, \ldots, y_p,z_1, \ldots, z_q]$ for some
  context $C$, or $r$ is a subterm of some~$l_i$ thus $\var(r)
  \subseteq \var(l_1, \ldots, l_n)$. \qed
\end{proof}

Consider the convergent  theories of blind signatures $\E_\blind$ and prefix
encryption~$\E_\pref$ defined by the following sets of equations.
\[
{\cal E}_{\blind} =  \left\{
\begin{array}{rcl}
\checksign(\sign(x,y),\pub(y)) &=& \okay\\
\unblind(\blind(x, y),y) & =& x\\
\unblind(\sign(\blind(x,y),z),y) & = & \sign(x,z)
\end{array} 
\right\}
\]
\[
{\cal E}_{\pref} = {\cal E}_\enc \cup \left\{
\begin{array}{rcl}
\pref(\enc(\langle x, y \rangle, z)) & =& \enc(x,z)
\end{array} 
\right\}
\]
The theory $\E_\blind$ models primitives used in e-voting
protocols~\cite{DKR-jcs08}.
The prefix theory represents the property of many chained modes of
encryption (e.g. CBC) where an attacker can retrieve any encrypted
prefix out of a ciphertext.

\begin{lemma}
The rewrite system  associated to 
the theory of homomorphism $\E_\homo$ defined in
Section~\ref{sec:eq} as well as the rewrite systems obtained by orienting from
left to right the equations in $\E_\blind$ and $\E_\pref$ are layered
convergent.
\end{lemma}

\begin{proof}
Let us check for instance that the prefix theory~$\E_\pref$ is layered. Let $N=1$, $\R_1$ be the rewrite system obtained from
${\cal E}_\enc$ by orienting the equations from left to right, and
$\R_2 = \R_1\cup\{\pref(\enc(\langle x, y \rangle, z)) \rightarrow
\enc(x,z)\}$.
The rewrite rules of~$\R_1$ satisfy the assumptions since $\R_1$
forms a convergent subterm rewrite system.
The additional rule $\pref(\enc(\langle x, y \rangle, z)) \rightarrow
\enc(x,z)$ admits three
decompositions up to permutation of parameters:
\begin{itemize}
\item $l = \pref(l_1)$, in which case $\var(r) \subseteq \var(l_1)$;
\item $l = \pref(\enc(l_1,z))$, in which case 
$\enc(\proj_1(l_1),z)\rightarrow_{\R_1} r$;
\item $l = \pref(\enc(\langle x, y \rangle,z))$, in which case 
$r = \enc(x,z)$.
\end{itemize}

Verifying that the convergent theories $\E_\homo$ and
$\E_\blind$ are layered is similar.
\end{proof}

%\commentaire{S.D.: souhaite-t-on mettre les preuves en annexe ?}

%
% NON-FAILURE
%

\subsection{A syntactic criterion}

\begin{definition}[Maximal]
We say that the function $\Ctx$ is \emph{maximal} if for every
  $\phi$ and $t$, if there exists $s$ such that $\phi \vdash s$ and $t
  \rR^* s$, then  $\Ctx(\phi \vdash^?_\R t) \neq \bot$.
\end{definition}

\begin{proposition}\label{prop:criterion}
  Assume that the function $\Ctx$ in use is \emph{maximal}. 
Then, provided that $\R$ is layered
  convergent, there exists no state $(\Phi,\Psi)$
  from which $(\Phi,\Psi) \Longrightarrow \bot$ is the only applicable
  derivation.
\end{proposition}

\begin{proof}
  By contradiction, let $(\Phi, \Psi)$ be a state from which $(\Phi,
  \Psi) \Longrightarrow \bot$ is the only applicable derivation, and
  let $l \to r$ be the rewrite rule involved in the corresponding
  instance of~\textbf{A.3}.
We prove the property by induction on the index $i \in \{0\ldots N\}$
  such that $l \to r \in \R_{i+1}-\R_{i}$.
Using the notations of Figure~\ref{fig:rules} for the instance of
\textbf{A.3} under consideration and the assumption on $\Ctx$,
we have that:
\begin{enumerate}[(a)]
\item for every $r\sigma \rR^* s$,\; $\Phi \cup \{z_1 \ded z_1, \ldots,
  z_q \ded z_q\} \not\vdash s$,\;and
\item $(r\sigma)\downR$ is not ground.
\end{enumerate}

In particular, (b) implies that $\var(r)$ is not included in
$\var(l_1,\ldots,l_n)$, otherwise we would have
\begin{eqnarray*}
\var((r\sigma)\downR) &\subseteq& \var(r\sigma) \subseteq
\var(\var(r)\sigma)\\ &\subseteq& \var(\var(l_1,\ldots,l_n)\sigma)
\subseteq \var(t_1,\ldots,t_n)=\emptyset
\end{eqnarray*}

By assumption on the decomposition $l = D[l_1, \ldots, l_n, y_1,
\ldots, y_p, z_1,\ldots, z_q]$ of $l \to r \in \R_{i+1}-\R_{i}$, we
deduce that there exists some contexts $C_0,\ldots,C_k$ and some
terms $s_1, \ldots, s_k$ such that:
\begin{itemize}
\item $r=C_0[s_1, \ldots, s_k]$;
\item for each $1 \leq i \leq k$, $C_i[l_1, \ldots, l_n, y_1,
  \ldots, y_p,z_1, \ldots, z_q]$ rewrites to $s_i$ in zero or one
  step of rewrite rule in head position along $\R_i$.
\end{itemize}
Let $C=C_0[C_1,\ldots,C_k]$ and $t_0=C[l_1, \ldots, l_n, y_1, \ldots,
y_p,z_1, \ldots, z_q]$. Note that $t_0 \to^*_{\R_i} r$.
If $t_0=r$, we obtain that $r\sigma = C[t_1, \ldots, t_{n+p}, z_1,
\ldots, z_q]$ is syntactically deducible from $\Phi \cup \{z_1 \ded
z_1, \ldots, z_q \ded z_q\}$, which contradicts (a). Hence $t_0
\to^+_{\R_i} r$, and in particular $i > 0$.

Let $\mu$ be a substitution mapping the variables $z_j$ to distinct
fresh public constants~$\cst a_j$. For each $1 \leq i \leq k$, let $u_i 
=C_i[l_1, \ldots, l_n, y_1, \ldots,
y_p,z_1, \ldots, z_q]\sigma\mu$ .
The term $u_i= C_i[t_1, \ldots, t_{n+p}, \cst a_1,
\ldots, \cst a_q]$ is syntactically deducible from $\Phi$, and reduces to
$u_i'=s_i\sigma\mu$ in zero or one step (in head position) along~$\R_i$.

By induction hypothesis on~$i-1$, no applicable rule \textbf{A.3} from
$(\Phi, \Psi)$ may involve a rule in $\R_i$. Besides, by assumption,
$(\Phi, \Psi)$ is saturated for the rules \textbf{B.1}, \textbf{B.2},
\textbf{A.1} and~\textbf{A.2}.
Therefore, Lemma~\ref{lem:comp} applied to $\Phi \vdash u_i$ and $u_i
\to_{\R_i} u_i'$ implies that there exists~$u_i''$ such that $u_i'\rR^* u_i''$
and $\Phi \vdash u_i''$. The same conclusion trivially holds if $u_i'=u_i$.
Let $s = C_0[u_1'', \ldots, u_k'']\mu^{-1}$ be the term obtained by
replacing each $\cst a_i$ by $z_i$ in $C[u_1'', \ldots, u_k'']$. Since the
$\cst a_i$ do not occur in $\R$ nor in $\Phi$, we deduce that~$s$ satisfies
$r\sigma = C_0[s_1\sigma, \ldots, s_k\sigma] = C_0[u_1', \ldots, u_k']\mu^{-1}
\rR^* s$ 
and $\Phi \cup \{z_1 \ded z_1, \ldots, z_q \ded
z_q\}\vdash s$, in contradiction with the condition (a) stated at the
beginning of the proof.\qed
\end{proof}

%
% FIN PREUVE NON-FAILURE
%

% \begin{comment}
%   M.B. Petit problème: la fonction ci-dessous (utilisée dans YAPA) ne
%   marche pas que dans les cas où $(\Phi, \Psi)$ est déjà saturé pour
%   $\R_0 \ldots \R_i$. Ca n'est pas grave, mais c'est un peu pénible à
%   expliquer.
% \end{comment}

\subsection{Practical considerations.}
\label{sec:termination-fail}

%\textcolor{blue}{S.D.: section a relire.}

Unfortunately, such a maximal $\Ctx$ is too inefficient in
practice as one has to consider the syntactic deducibility problem $\phi
 \vdash s$ for every $t \rR^* s$. Proposition~\ref{pro:maxequivsimple} below shows that the
 simple function context is actually sufficient to ensure non-failure when we
 know that another function $\Ctx$ already prevents failure on any state (reachable or
 not).

\begin{proposition}
\label{pro:maxequivsimple}
Let $\R$ be a convergent rewrite system and $\Ctx_0$ be an arbitrary function
$\Ctx$. If there exists no state
$(\Phi,\Psi)$ from which $(\Phi,\Psi)\Longrightarrow \bot$ is the only
applicable derivation when the function $\Ctx$ in use is $\Ctx_0$, then
there  exists no state
$(\Phi,\Psi)$ from which $(\Phi,\Psi)\Longrightarrow \bot$ is the only
applicable derivation for any choice of~$\Ctx$.
\end{proposition}

\begin{proof}
Let $\Ctx_0$ and $\Ctx'_0$ be two arbitrary functions $\Ctx$  ({\it i.e.} they satisfy
properties~(\ref{item:ctx_sufficient}) and~(\ref{item:ctx_correct})).
Assume that there exists no state $(\Phi, \Psi)$ from which  $(\Phi,\Psi)\Longrightarrow \bot$ is the only
applicable derivation when the function $\Ctx$ in use is $\Ctx_0$. Assume by
contradiction that there exists a state $(\Phi_0, \Psi_0)$ from which 
$(\Phi_0,\Psi_0)\Longrightarrow \bot$ is the only
applicable derivation for~$\Ctx'_0$. 
This means that there exist:
\begin{itemize}
\item a rewrite rule $l \to r \in \R$,
\item  a proper decomposition
$D[l_1,\ldots, l_n,y_1,\ldots,y_p, z_1, \ldots,z_q]$ of $l$,
\item some deduction facts $M_1 \ded t_1, \ldots, M_{n+p} \ded t_{n+p} \in
  \Phi_0$, and
\item a substitution $\sigma$ such that $(l_1, \ldots, l_n,y_1,\ldots, y_p)\sigma = (t_1,\ldots, t_{n+p})$. 
\end{itemize}
Moreover, since this instance corresponds to an instance of \textbf{A.3}, 
we have that $r\sigma\downR$ is not ground. When the function $\Ctx$ in use is
$\Ctx_0$, this instance has to correspond to an instance of \textbf{A.1} 
(\textbf{A.2} and \textbf{A.3} are impossible).
Hence, we have that $\Ctx_0(\Phi_0 \cup \{z_1 \ded z_1, \ldots, z_q \ded z_q\}
\vdash_{\R}^{?} r\sigma) \neq \bot$. This means that there exists~$s$ 
such that $r\sigma \to_{\R}^* s$ and
$\Phi_0 \cup \{z_1 \ded z_1, \ldots, z_q \ded z_q\} \vdash s$. 
Since $\R$ is convergent, we have that $s \to_{\R}^* r\sigma\downR$.

Let $\mu$ be a substitution mapping the variables $z_j$ to distinct fresh
public constants~${\cst a_j}$. We have that
$s\mu \to_{\R}^* (r\sigma\downR)\mu$ and also that $\Phi_0 \vdash s\mu$.
Since $(\Phi_0,\Psi_0)\Longrightarrow \bot$ is the only
applicable derivation for~$\Ctx'_0$, the rules \textbf{A.2}, \textbf{B.1}, and \textbf{B.2} cannot be applicable, even for~$\Ctx_0$. We saturate $(\Phi_0,\Psi_0)$ with the $\textbf{A.1}$ rule for~$\Ctx_0$, reaching a state of the form $(\Phi_0,\Psi_0')$ since only equations can be added to the state. Note also that the $\textbf{A.1}$ rule can only be applied a finite a number of time and does not trigger the other rules. Thus $(\Phi_0,\Psi_0')$ is saturated for~$\Ctx_0$. 
Using  Lemma~\ref{lem:comp} (with the function $\Ctx_0$), we obtain
that $\Phi_0 \ded (r\sigma\downR)\mu$, and thus $\Phi_0 \cup \{z_1 \ded z_1,
\ldots, z_q \ded z_q\} \vdash r\sigma\downR$.
This contradicts the fact that \textbf{A.1}
does not apply  on $(\Phi_0, \Psi_0)$ when the function $\Ctx$ in use is
$\Ctx'_0$. Hence, the result.
\qed
\end{proof}

\begin{corollary}
Let $\R$ be a layered convergent rewrite system and consider an arbitrary
function $\Ctx$ in use. There exists no state
$(\Phi,\Psi)$ from which $(\Phi,\Psi) \Longrightarrow \bot$ is the only
applicable derivation.
\end{corollary}

\section{Termination}
\label{sec:termination}
In the previous section, we have described a sufficient criterion for
non-failure. As shown by the example given below, 
this criterion does not ensure the termination of our saturation
procedure.

\begin{example}
Consider the following layered convergent rewrite system $\f(\g(x)) \to \g(\h(x))$ where $\f$ is
a public function symbol whereas $\g$ and $\h$ are private function symbols.
Let $\varphi = \{\w_0 \ded \g(\cst a)\}$ where $\cst a$ is a private constant.
By repeatedly applying the \textbf{A} rule on the newly generated deduction
fact, we generate an infinite number of deduction facts of the form:
\[
\f(\w_0) \ded \g(\h(a)),  \; \f(\f(\w_0)) \ded \g(\h(\h(a)), \;
\f(\f(\f(\w_0))) \ded \g(\h(\h(\h(a))), \; \ldots
\]
\end{example}

To obtain decidability for a
given layered convergent theory, there remains only to provide a 
termination argument.
Such an argument is generally easy to develop by hand as we illustrate
on the example of the prefix theory. For the case of existing
decidability results from~\cite{AbadiCortierTCS06}, such as the
theories of blind signature and homomorphic encryption, we also
provide a semantic criterion that allows us to directly conclude
termination of the procedure.
Note that this semantic criterion does not apply only to layered convergent theories but to any convergent theories (for which failure is guaranteed not to happen).
% %
% In brief, for a given initial frame $\varphi$, the criterion states
% that there exists another frame $\varphi_s$ such that any $\R$-reduced
% term is deducible from $\varphi$ modulo~$\E$ if and only if it is
% syntactically deducible from~$\varphi_s$. Actually, this condition is
% necessary and sufficient for our procedure to terminate without
% failure, provided that only strict instances of rule \textbf{A.3} and
% a \emph{fair} strategy for rule application are used.

\subsection{Termination of \textbf{B}~rules}

To begin with, we note that 
\textbf{B}~rules always terminate after a polynomial number of
steps. Let us write $\Longrightdotarrow^n$ for the relation made of
exactly $n$ \emph{strict applications} of rules ($S \Longrightdotarrow
S'$ iff $S \Longrightarrow S'$ and $S \neq S'$).

\begin{proposition} \label{prop:B_termination} For every states
  $S=(\Phi,\Psi)$ and~$S'$ such that
  $S\Longrightdotarrow^n S'$ using only \textbf{B}
  rules, $n$ is polynomially bounded in the %DAG-
size of $\im(\Phi)$.
\end{proposition}
%\begin{comment}
%  M.B. En fait, plus précisément, je dirais qu'on a $n_{B_2} \leq s$ et
%  $n_{B_1} \leq s^2$, d'où $n \leq s^2 + s$.
% S.D.: ca fait un peu bizarre de mettre DAG dans une proposition puisque nous
% ne souhaitons pas definir cette notion.
%
% V.C. : j'ai viré le mot DAG car je trouve que ca fait bizarre en
% effet. Je laisse les DAG ensuite là où c'est nécessaire (mais mieux
% expliqué). Vu qu'on ne détaille pas trop la complexité, ca me parait
% suffisant comme ca.
%
% M.B. ah bon, ok.
%\end{comment}
This is due to the fact that frames are one-to-one and that the rule
\textbf{B.2} only adds deduction facts $M\ded t$ such that $t$ is a
subterm of an existing term in 
$\Phi$.

\subsection{Proving termination by hand.} 
% To obtain decidability for a
% given layered convergent theory, it remains only to provide a 
% termination argument.
%

%Hence, 
For proving termination, we observe that it is sufficient
to provide a function $s$ mapping each frame $\Phi$ to a finite set of
terms $s(\Phi)$ including the subterms of $\im(\Phi)$ and
such that rule~\textbf{A.2} only adds deduction facts $M\ded t$
satisfying $t\in s(\Phi)$.

\medskip{}

For subterm theories, we obtain polynomial termination by choosing
$s(\Phi)$ to be the subterms of $\im(\Phi)$ together with the
ground right-hand sides of $\R$.
%% For the particular case of convergent subterm theories, it is easy to
%% notice that both \textbf{A} and \textbf{B} rules only add deduction
%% facts $M\ded t$ where $t$ is a subterm of the initial frame.
%Hence the procedure saturates after a polynomial number of
%applications of rules.

\begin{proposition}
  Let $\E$ be a %(fixed)
weakly subterm convergent theory. For every %states
  \mbox{$S=(\Phi,\Psi)$} and $S'$ such that $S\Longrightdotarrow^n S'$, $n$ is
  polynomially bounded in the %DAG-
size of $\im(\Phi)$.
\end{proposition}

To conclude that deduction and static equivalence are decidable in
polynomial time~\cite{AbadiCortierTCS06}, we need to show that the
deduction facts and the equations are of polynomial size. This
requires a DAG representation for terms and visible equations. For our
implementation, we have chosen not to use DAGs for the sake of simplicity 
since DAGs require much heavier data structures. However,
similar techniques as those described in~\cite{AbadiCortierTCS06}
would apply to implement our procedure using DAGs.

\smallskip{}

For proving termination for the prefix theory~$\E_\pref$, it suffices
to consider $s(\phi) = \stext(\Phi)$,
where the notion of extended subterm is recursively defined as
follows:
\begin{itemize}
\item $\stext(a) =  \{a\} \;\;\; \mbox{ if $a$ is a constant or a variable}$
\item $\stext(f(t_1,\ldots,t_n))  =  \{f(t_1,\ldots,t_n)\}\cup\bigcup_{i=1}^n{\stext(t_i)}
\;\, f\in\{\dec,\langle,\rangle,\proj_1,\proj_2,\pref\}$
\item $\stext(\enc(t,u))  =  \{\enc(t,u),\enc(t_1,u)\}\cup\stext(t)\cup\stext(u)
\;\;\; \mbox{ if }t = \langle t_1,t_2\rangle$
\item $\stext(\enc(t,u)) = \{\enc(t,u)\}\cup\stext(t)\cup\stext(u)
\;\;\; \mbox{ otherwise}$.
\end{itemize}

\begin{proposition}
Consider the prefix theory~$\E_\pref$. 
For every \mbox{$S=(\Phi,\Psi)$} and~$S'$ such that $S\Longrightdotarrow^n
S'$, $n$ is polynomially bounded in the 
size of $\im(\Phi)$.
\end{proposition}

We then deduce that deduction
and static equivalence are decidable for the equational theory
$\E_\pref$, which is a new decidability result.

\begin{corollary} 
Deduction and static equivalence are decidable in polynomial time for the equational theory
$\E_\pref$.
\end{corollary}

Similarly, we may retrieve decidability 
of deduction and static
equivalence for~$\E_\homo$ and~$\E_\blind$. However, we provide another
criterion that allows one to derive these facts from existing results.

%\medskip{}

\subsection{A semantic criterion} We now provide a
semantic criterion that more generally explains why our procedure
succeeds on theories previously known to be
decidable~\cite{AbadiCortierTCS06}. This criterion intuitively states that the set
of deducible terms from any initial frame $\varphi$
should be equivalent to a set of \emph{syntactically} deducible terms.
Provided that failures are prevented and assuming a \emph{fair}
strategy for rule application, we prove that this criterion is a
necessary and sufficient condition for our procedure to terminate.

%\textcolor{blue}{S.D.: c'est un peu genant de mettre ... dans la definition de
%  same instance, non ? Que faut-il mettre en plus pour completer ?}
\begin{definition}[Fair derivation]\label{def:fair}
  An infinite derivation 
\[(\Phi_0, \Psi_0) 
%\Longrightarrow  (\Phi_1, \Psi_1) 
\Longrightarrow \ldots \Longrightarrow (\Phi_n,
  \Psi_n) \Longrightarrow \ldots
\]
  is \emph{fair} iff along this derivation,
  \begin{enumerate}[(a)]
  \item \textbf{B} rules are applied with greatest priority, and
  \item whenever a \textbf{A} rule is applicable for some instance
    $(l \to r, D, t_1,\ldots,t_n,\ldots)$, eventually the same
    instance of rule is applied during the derivation.
  \end{enumerate}
\end{definition}

Fairness implies that any deducible term is eventually syntactically
deducible. This result follows from Lemma~\ref{lem:compbase} and Lemma~\ref{lem:comp}.

%\textcolor{blue}{S.D.: a relire.}
\begin{lemma}
\label{lem:ded}
Let $S_0 = (\Phi_0, \Psi_0) \Longrightarrow \ldots \Longrightarrow
(\Phi_n, \Psi_n) \Longrightarrow \ldots$ be an infinite fair
derivation from a state $S_0$. For every ground term $t$ such that $\Phi_0
\vdash_\E t$, either $(\Phi_0, \Psi_0) \Longrightarrow^*
\bot$ or there exists $i$ such that $\Phi_i \vdash t\downR$.
\end{lemma}

\begin{proof}
Let $t$ be a ground term deducible from $\Phi_i$ modulo~$\E$.
There exists $t_0$ such that $M \ded_{\Phi_i} t_0$ and $t_0 \to^* t\downR$.
This means that there exist a (public) context $C$ and some deduction facts 
$M_1\ded t_1,\ldots, M_n\ded t_n\in\Phi_i$ such that $M = C[M_1, \ldots, M_n]$
and $t_0 = C[t_1, \ldots, t_n]$. 

We show that 
either $(\Phi_i, \Psi_i)
\Longrightarrow^* \bot$ or
there exists $j \geq i$ such that $t\downR$ is
syntactically deducible from $\Phi_j$, 
by induction on $t_0$ equipped with the order
$<$ induced by the rewriting relation (that is $t_1<t_2$ if and only if
$t_2\to^+ t_1$).

\noindent\emph{Base case: $t_0 = t\downR$.} 
In such a case, since $\Phi_i \vdash t_0$, 
we have that $\Phi_i \vdash t\downR$. This allows us to conclude.

\smallskip{}
\noindent\emph{Induction step: $t_0 \to t' \to^* t\downR$.}
% Since $t_0 \to t'$, there exist a position $\alpha$, a substitution $\sigma$
% and a rewrite rule $l \to r \in \R$ such that $t_0|_{\alpha} = l\sigma$ and
% $t'= t_0[r\sigma]_{\alpha}$. 
%We note that $\alpha$ must be a (symbol) position
% of $C$ since each $t_i$ ($1 \leq i \leq n$) is $\R$-reduced. Hence, we may
% write $C|_{\alpha}[t_1,\ldots, t_n] = l\sigma$.
% Thus, the rewriting step mentionned above corresponds to a proper
% $(n,p,q)$-decomposition $D$ of $l$: $l = D[l_1, \ldots, l_n,y_1, \ldots, y_p,
% z_1,\ldots, z_q]$. 

Along a fair derivation, \textbf{B} rules are applied in priority. Hence, 
we choose the smallest
$i_1 \geq i$ such that no more \textbf{B} rules can be applied from $(\Phi_{i_1},
\Psi_{i_1})$.
Note indeed that there is no infinite derivation with only \textbf{B}
rules (Proposition~\ref{prop:B_termination}).
We have still that $C[M_1,\ldots,M_n]\ded_{\Phi_{i_1}} t_0\to t'$.

Applying Lemma~\ref{lem:comp} and observing that no \textbf{B} rule
can be applied from $(\Phi_{i_1}, \Psi_{i_1})$, 
we are in one of the following cases:
\begin{itemize}
\item $(\Phi_{i_1}, \Psi_{i_1}) \Longrightarrow \bot$. In such a case, we
  easily conclude since $(\Phi_0,\Psi_0) \Longrightarrow^* \bot$.
\item $\Phi_{i_1} \vdash t''$ for some $t''$ such that $t'\to^*_{\R} t''$. In
  such a case, we conclude by applying our induction hypothesis since $t''<
  t'< t_0$. There exists $j \geq i_1$ such that $\Phi_j \vdash t\downR$.
\item Otherwise an instance ($l \to r$, $D$, $t_1, \ldots, t_n$, \ldots) of a \textbf{A}
rule is applicable. Note that this instance is entirely determined by the rewrite rule
$l \to r$ involved in the rewriting step $t_0 \to t'$, the deduction facts
$M_i \ded t_i$ ($1 \leq i \leq n$) and the public context that witness the
fact that $\Phi_i \vdash t_0$.
\end{itemize}

By fairness, we know that a \textbf{A}
  rule will be applied along the derivation for the same instance  ($l \to r$, $D$, $t_1, \ldots,
  t_n$, \ldots).  Let $i_2$ be the indice on which this instance is applied.
  We have that $i_2 \geq i_1$.
Note that since \textbf{B} rules are applied in
  priority, $(\Phi_{i_2}, \Psi_{i_2})$ is saturated for \textbf{B} rules.
Either, we have that 
 $(\Phi_{i_2}, \Psi_{i_2}) \Longrightarrow \bot$ (and thus $(\Phi_i, \Psi_i)
 \Longrightarrow^* \bot$) or
  $(\Phi_{i_2}, \Psi_{i_2}) \Longrightarrow (\Phi_{{i_2}+1}, \Psi_{{i_2}+1})$.

We have that $C[M_1, \ldots, M_n] \ded_{\Phi_{i_2}} t_0$ and $t_0 \to_\R t'$.
By Lemma~\ref{lem:comp}, either $(\Phi_{i_2}, \Psi_{i_2}) \Longrightarrow \bot$ or
there exists $(\Phi'_{i_2}, \Psi'_{i_2})$, $M'$ and $t''$ such that:
\begin{itemize}
\item  $(\Phi_{i_2}, \Psi_{i_2}) \Longrightarrow (\Phi'_{i_2}, \Psi'_{i_2})$;
\item $M'\ded_{\Phi'_{i_2}} t''$ with $t'\to^*_{\R} t''$; and
\item $\Psi'_{i_2} \models C[M_1, \ldots, M_n] \bowtie M'$.
\end{itemize}
Actually, the instance of the \textbf{A} rule that is applied in this
derivation is entirely determined by the rewrite rule $l \to r$ involved in
the rewriting step $t_0 \to t'$,  the public context $C$ and the deduction
facts $M_i \ded t_i$ ($1 \leq i \leq n$) that witness the fact that $\Phi_i
\vdash t_0$ (and thus $\Phi_{i_2} \vdash t_0$). Hence, we have
that $(\Phi'_{i_2}, \Psi'_{i_2}) = (\Phi_{{i_2}+1}, \Psi_{{i_2}+1})$.

Thus we have that $M'\ded_{\Phi'_{{i_2}+1}} t''$ with $t'' \to^* t\downR$ and $t''
< t' <t$. We can apply our induction hypothesis, either  $(\Phi_{{i_2}+1},
\Psi_{{i_2}+1}) \Longrightarrow^* \bot$ (and thus $(\Phi_i, \Psi_i)
\Longrightarrow^* \bot$) or there exists $j \geq {i_2}+1$ such that $\Phi_j
\vdash t\downR$.
\end{proof}

Our termination criteria  (Property~$(ii)$ below) is a semantic criterion.
It is related to the notion \emph{locally stable} introduced in~\cite{AbadiCortierTCS06}. 

\begin{proposition}[criterion for termination]
\label{theo:termination}
  Let $\varphi$ be an initial frame such that $\Init(\varphi)
  \,\, \not \!\!\Longrightarrow^* \bot$.
  The following conditions are equivalent:
  \begin{enumerate}[(i)]
  \item There exists a     saturated couple $(\Phi, \Psi)$ such that
    $\Init(\varphi) \Longrightarrow^* (\Phi, \Psi)$.
  \item There exists a (finite) initial frame $\varphi_s$ such that
    for every term $t$, $t$ is deducible from $\varphi$ modulo $\E$ iff
    $t\downR$ is syntactically deducible from $\varphi_s$.
 \item  There exists no fair infinite derivation
    starting from $\Init(\varphi)$.
  \end{enumerate}
\end{proposition}

\begin{proof}
$(iii) \Rightarrow (i)$: trivial. Indeed by using a fair derivation we will
  eventually reach a weakly saturated state.
$(i) \Rightarrow (ii)$: Let $\Phi = \{M_1 \ded s_1, \ldots, M_\ell \ded
  s_\ell\}$ and  $\varphi_s = \{\w_1 \ded s_1, \ldots, \w_\ell \ded s_\ell\}$.
  Let $t$ be a ground term. By
  Theorem~\ref{theo:soundcomp}, we have that $\exists M \, . \, M \ded_\varphi^\E t$
  iff $\exists M \, .\, M \ded_\Phi t\downR$, i.e. $\exists M \, .\, M
  \ded_{\varphi_s} t\downR$.
$(ii) \Rightarrow (iii)$: we need to prove that
there exists no fair
    infinite derivation starting from $\Init(\varphi)$.

Let $\varphi_s = \{\w_1 \ded s_1,
    \ldots, \w_\ell \ded s_\ell\}$ an initial frame
 such that for every~$t$,  $\exists M \, .\,
    M \ded_\varphi^\E t$ is
    equivalent to $\exists M \, . \,M \ded_{\varphi_s} t\downR$.
Assume by contradiction that there is an infinite fair derivation
$(\Phi_0, \Psi_0) 
\Longrightarrow \ldots \Longrightarrow (\Phi_n,
  \Psi_n) \Longrightarrow \ldots$ with $(\Phi_0, \Psi_0)  = \Init(\varphi)$.

By Lemma~\ref{lem:ded} and since $\Init(\varphi)
    \,\,\not\!\!\Longrightarrow^* \bot$, we deduce that 
there exists $i_0$ such that each $s_i$, $1\leq i\leq \ell$ is
syntactically deducible from $\Phi_{i_0}$.
Since there is no infinite derivation with only \textbf{B} rules (Proposition~\ref{prop:B_termination}), we
    can also assume that no \textbf{B} rule can be applied from $\Phi_{i_0}$.
We have that $\exists M \, .\,
    M \ded_\varphi^\E t$ is
 now   equivalent to $\exists M \, . \,M \ded_{\Phi_{i_0}} t\downR$
    thus the \textbf{A.2} rule cannot be applied either.
We deduce that no deduction facts are added to $\Phi_{i_0}$ along the
    derivation, that is $\Phi_{j} = \Phi_{i_0}$ for every $j\geq
    i_0$.
Since no deduction fact are added, only a finite number of 
 \textbf{A.1} rules can be applied, which contradicts the existence of
    an infinite chain. \qed
\end{proof}

% Interestingly, no assumption is made on the function~$\Red$ involved
% in the \textbf{A}~rules. We deduce that the choice of $\Red$ does not
% affect saturation. %the success of the procedure.
% %
% The most difficult part of the proof concerns the implication $(ii)
% \Rightarrow (iii)$.  Non-failure is proved by contradiction: 
% from any instance of rule \textbf{A.3} %  while the other
% % rules are not (strictly) applicable, %%M.B. inutile c.f. preuve!
% we provide a way to construct infinitely many deducible terms that are
% not public contexts over finitely many deducible terms. As for
% termination, we first apply Proposition~\ref{prop:ded} and obtain that
% the terms $s_1$, \ldots, $s_\ell$ given by Property $(ii)$ eventually become all
% syntactically deducible during saturation. We then deduce
% that no more deduction facts can be added, and conclude termination
% of the derivation.

Together with the syntactic criterion described in
Section~\ref{sec:fail} to prevent non-failure, 
this criterion (Property~$(ii)$) allows us to prove 
decidability of deduction
and static equivalence for layered convergent theories
that belong to the class of {locally stable} theories defined 
in~\cite{AbadiCortierTCS06}.  As a consequence, our procedure always
saturates for the theories of blind signatures and homomorphic
encryption since those theories are layered and have been proved
locally stable~\cite{AbadiCortierTCS06}.  Other examples of layered
convergent theories enjoying this criterion can be found
in~\cite{AbadiCortierTCS06} (e.g. a theory of addition).
While in~\cite{AbadiCortierTCS06} the decision algorithm needs 
to be adapted for each theory, we propose a single (and efficient) 
algorithm that ensures a unified treatment of all these theories.

%%% Local Variables: 
%%% mode: latex
%%% TeX-master: "main"
%%% End: 

\section{Implementation: the tool YAPA} %Mathieu: pour moi les noms propres ou les accronymes d'outil ne sont pas des adjectifs: ``the YAPA tool'' => ``the tool YAPA'' ou bien ``YAPA'' tout court.
\label{sec:experiment}

YAPA (Yet Another Protocol Analyzer) is an Ocaml implementation of the
saturation procedure presented in Section~\ref{sec:rules} with several optional optimizations.
It can be freely downloaded\footnote{\url{http://www.lsv.ens-cachan.fr/~baudet/yapa/index.html}}
together with a brief manual and examples.

The tool takes as input an equational theory described by a finite
convergent rewrite system, as well as frame definitions and
queries.
% Correctness and completeness are guaranteed by
% Theorem~\ref{theo:soundcomp}.
%  It terminates without failure in many cases as shown by
%  Theorem~\ref{theo:termination}.
The procedure starts by computing the decompositions of the
rewrite system.
By default, the following optimization is done: provided that the
rewrite rules are given in an order compatible with the sets $\R_0
\subseteq \ldots \subseteq \R_{N+1}$ of Definition~\ref{def:layered},
the tool is able to recognize layered theories and to
pre-compute the associated contexts~$C$ related to condition~(ii) of
this definition. This allows resolving the failure cases as soon as they appear, rather than later on, when the saturation procedure has made enough progress.
This optimization was studied in a first version of this article~\cite{BCD-RTA09} but as the practical benefits appear to be minor (see below), we chose not to keep these technical developments in this version for the sake of notational simplicity.

Another optimization concerns a specific treatment of subterm
convergent theories but does not induce any difference with the
theoretical procedure presented here.
Except for the first (optional) optimization mentioned above, the algorithm follows the procedure described in Section~\ref{sec:rules}, using a minimal function $\Ctx$ in the sense in Section~\ref{sec:termination-fail}, and a fair strategy of rule application (see Definition~\ref{def:fair}).

\medskip{}

We have conducted several experiments on a PC Intel Core 2 Duo at
2.4~GHz with 2~Go RAM 
for various equational theories (see below) and
found that YAPA provides an efficient way to check static equivalence
and deducibility. Those examples are available at
\url{http://www.lsv.ens-cachan.fr/~baudet/yapa/index.html}.
The figures given below are valid for the versions with and without optimizations. %TODO: Stephanie c'est vrai ???

For the case of $\E_\enc$,
% in order to evaluate the performance of the tool,
we have run YAPA on the frames:
\begin{itemize}
\item  $\varphi_n = \{\w_1\ded t_n^0,\w_2\ded
\cst c_0,\w_3\ded \cst c_1\}$, and 
\item $\varphi_n'=\{\w_1\ded t_n^1,\w_2\ded \cst c_0,\w_3\ded
\cst c_1\}$, 
\end{itemize}
\noindent where $t_0^i = \cst c_i$ and $t_{n+1}^i =
\langle\enc(t_n^i,\cst k_n^i),\cst k_n^i\rangle$, $i\in\{0,1\}$. These examples
allow us to increase the (tree, non-DAG) size of the distinguishing
tests exponentially, while the sizes of the frames grow
linearly. Despite the size of the output% and the absence of DAG support in
% YAPA
, we have observed satisfactory performances for the tool.

\begin{center}
\begin{tabular}{|c|c|c|c|c|c|c|c|c|c|c|}
\hline 
\begin{tabular}{c}Equational \\theory \end{tabular}& 
\begin{tabular}{c}$\E_\enc$\\$n=10$\end{tabular}&
%\begin{tabular}{c}$\E_\enc$\\$n=12$\end{tabular}&
\begin{tabular}{c}$\E_\enc$\\$n=14$\end{tabular}&
\begin{tabular}{c}$\E_\enc$\\$n=16$\end{tabular}&
\begin{tabular}{c}$\E_\enc$\\$n=18$\end{tabular}&
\begin{tabular}{c}$\E_\enc$\\$n=20$\end{tabular}
%$\E_\blind$ &
%$\E_\pref$ &
%$\E_\homo$ &
%$\E_\add$ 
\\\hline
Execution time &
$<$ 1s  &
%1,2s &
1,7s &
 8s &
30s &
$<$ 3min 
%$<$ 1s  &
%$<$ 1s  &
%$<$ 1s  &
%$<$ 1s 
\\\hline
\end{tabular}
\end{center}

We have also experimented YAPA on several convergent theories, e.g. $\E_\blind$, $\E_\homo$,
$\E_\pref$ and the theory of addition $\E_\add$ defined
in~\cite{AbadiCortierTCS06}.
% \begin{comment2}
%   M.B. Pour être complet, il faudrait décrire les frames en annexe ou
%   (mieux?) donner un pointeur sur les fichiers. De toute façon, je
%   dois faire une nouvelle release de YAPA intégrant les derniers changements...
% \end{comment2}

\paragraph*{Comparison with ProVerif}

In comparison with the tool ProVerif~\cite{BlanchetCSFW01,BlanchetAbadiFournetJLAP07}, here
instrumented to check static equivalences, our test samples suggest a
running time between one and two orders of magnitude faster for
YAPA. Also we did not succeed in making ProVerif terminate on the two
theories $\E_\homo$ and $\E_\add$. Of course, these results are not entirely surprising given that
ProVerif is tailored for the more general (and difficult) problem of
protocol (in)security under active adversaries. In particular
ProVerif's initial preprocessing of the rewrite system appears more
substantial 
than ours and does not terminate on the theories $\E_\homo$ and $\E_\add$
(although termination is guaranteed for linear or subterm-convergent theories~\cite{BlanchetAbadiFournetJLAP07}).

% YAPA is usually faster than ProVerif~\cite{BlanchetCSFW01}, the
% only other tool that can check static equivalence. 
% ProVerif fails for the two theories $\E_\homo$ and $\E_\add$.
% However, it is important to notice that ProVerif is scaled for
% checking security also in the active case, which is out of the scope of the
% YAPA tool. 
% Another noticeable difference between YAPA and ProVerif is
% that there is no guarantee of termination for ProVerif, while YAPA's
% termination is guaranteed in many cases as shown is Section~\ref{sec:termination-fail}.

\paragraph*{Comparison with KiSs.}
The tool KiSs (Knowledge in Security protocolS) is  a C++ implementation of the procedure described
in~\cite{CDK-cade09}. This procedure reused the same
concepts than the one presented in a preliminary version of this work~\cite{BCD-RTA09}.
The performances of the tool YAPA are comparable to the performances of
KiSs. However, since the tool KiSs implements DAG representations for
terms, it does better on the example developed above.
%\begin{commentaire}
%Mathieu: ça serait bien d'être un peu plus précis sur les théories qui
%fonctionnent en pratique dans Kiss et celles qui bénéficient d'une preuve.
%Que signifie ``deal with'' et ``consider'' plus bas ?
% Stef: C'est pour parler des theories qui marchent dans chacun des outils.
% Je ne vois pas ce que tu veux mettre pour etre plus precis. Dans Kiss, la
% terminaison est prouvee pour sous-terme convergent plus quelques theories
% particulieres dont trapdoor.
%\end{commentaire}
From the point of view of the equational theories the tools are able to deal
with, they are incomparable.
KiSs allows one to consider some equational theories for which our procedure
fails (e.g. the theory of trapdoor bit commitment). 

Conversely our procedure is guaranteed to terminate (without failure) for
%classes of %Mathieu: commenté car on n'a pas vraiment de critère générique de terminaison
theories that are not considered by the procedure implemented in KiSS.
The only general class of theory for which KiSs has been proved to terminate
is the class of subterm convergent equational theory.

\section{Conclusion and future work}
\label{sec:conclu}

We have proposed a procedure for checking deducibility and 
static equivalence. Our procedure is correct and complete for 
any convergent theory and is efficient, as shown by its implementation 
within the tool YAPA. Since deducibility and static equivalence 
are undecidable in general, our algorithm may fail or may not terminate.
We have identified a large class of equational theories 
(called layered convergent) for which non-failure of the procedure 
is ensured. Since termination can then often be easily proved by hand, 
we have obtained a new decidability result for the prefix theory.
We have also proposed a semantic (and exact) characterization for 
the procedure to terminate. This again yields a new decidability 
result for locally stable, layered convergent theories.

As further work, we would like to extend our procedure to theories with 
associative and commutative operators. A first possibility would be to 
implement the decidability result of~\cite{CortierDelaune-LPAR07-monoidal} 
for monoidal theories (that include many theories with 
associative and commutative operators) and to combine the two procedures 
using the combination theorem of~\cite{ACD-frocos07}.
However, it seems much more efficient to integrated associativity and 
commutativity directly and this could even open the way to a more powerful combination technique.

The tool KiSS, developed recently~\cite{CDK-cade09}, supports %one to consider 
several equational theories for which our procedure fails. Conversely our 
procedure is guaranteed to terminate (without failure) for classes of theories 
that are not considered by the procedure implemented in KiSS.
It would be interesting to compare the techniques and possibly 
to combine them in order to capture more theories.

% \commentaire{TO DO, V.C. Pistes de future work:
% \begin{itemize}
% \item parallel key search (euh, il faudrait que je me rememore ce que c'est exactement)
% \item symbol AC ?
% \item une procedure qui engloberait YAPA et Kiss.
% \item Cominaison de procedures ?
% \end{itemize}
% }

\bibliographystyle{acmtrans}
\bibliography{biblio}
% \begin{received}
% Received February 1986;
% November 1993;
% accepted January 1996
% \end{received}

\newpage\appendix

\section{Appendix}

% A first lemma justifies the expression ``saturation procedure'' used
% throughout this paper. Its proof is done by induction on the
% derivation and by a straightforward inspection of the rules.

% Before going through the proofs, we note the following fact that justifies
% the expression ``saturation procedure'' used throughout this paper.
% \begin{lemma}[monotony]\label{lem:monotony}
%   Let $(\Phi,\Psi)$ and $(\Phi', \Psi')$ be such that
%   $(\Phi, \Psi) \Longrightarrow^* (\Phi', \Psi')$. We have that $\Phi
%   \subseteq \Phi'$ and $\Psi \subseteq \Psi'$.  In particular, we note
%   that
%   \begin{itemize}
%   \item $M \ded_\Phi t$ implies $M \ded_{\Phi'} t$;
%   \item $\Psi \models M \bowtie N$ implies  $\Psi' \models M \bowtie N$.
%   \end{itemize}
% \end{lemma}

% The next lemma states a number of invariants concerning reachables states.
% Again, its proof is done by induction on the derivation and by a
% straightforward inspection of the rules.

% \begin{lemma}[invariants]\label{lem:invariants}
% Let $\varphi$ be an initial frame and $(\Phi, \Psi)$ be 
% such that  $\Init(\varphi)\Longrightarrow^* (\Phi, \Psi)$. We
% have that
% \begin{enumerate}[(i)]
% \item $(\Phi, \Psi)$ is a state, that is, $\Phi$ is one-to-one, and for every
%   $M \ded t$ in $\Phi$, $t$ is $\R$-reduced;
% \item $\param(\Psi) \subseteq \param(\Phi) = \dom(\varphi)$.
% \end{enumerate}
% \end{lemma}

% We state these two lemmas in order to highlight some features of our
% saturation procedure. Actually, they are use 
% several times (sometimes implicitly) to
% establish soundness and completeness of the saturation procedure. 

\label{app:app-comp}

%%%%%%%%%%%%%%%%%%%%%%%%%%%%%%%%%%%%%%%%%%%%%%%%%%%%%%%%%%%%%%%%%%%
%                                                                 % 
%                         Syntactic Case                          %
%                                                                 %
%%%%%%%%%%%%%%%%%%%%%%%%%%%%%%%%%%%%%%%%%%%%%%%%%%%%%%%%%%%%%%%%%%%

\noindent\usebox{\lemcompsyntded}

\begin{proof}
By hypothesis, we have that $N \ded_{\Phi} t$. This means that there exists a
public context~$C$ and some facts $M_1 \ded t_1, \ldots, M_n \ded t_n \in
\Phi$ such that $N = C[M_1, \ldots, M_n]$ and $t = C[t_1,\ldots, t_n]$.
Let~$C$ be such a context whose size is minimal. We
show the result by structural induction on~$C$.

\smallskip{}

\noindent{\emph{Base case:}} $C$ is reduced to an
hole.  Let $(\Phi',\Psi') =
(\Phi, \Psi)$ and $N'= N$. The result trivially holds.

\smallskip{}

\noindent{\emph{Induction step:}} $C = f(C_1,\ldots, C_r)$ with $f \in \F_{\pub}$ of arity~$r$.
In such a case, we have ${t = f(u_1,\ldots, u_r)}$ and 
$C_i[M_1,\ldots, M_n] \ded_{\Phi} u_i$ with $u_i \in \st(t_0)$
for each $1 \leq i \leq r$. Thus, we can apply our induction hypothesis. We
deduce that there exists $(\Phi_1, \Psi_1)$ and terms $N'_1, \ldots N'_r$ such
that:

\begin{itemize}
\item $(\Phi, \Psi) \Longrightarrow^* (\Phi_1, \Psi_1)$ using $\mathbf{B}$ rules, 
\item  $N'_i\ded u_i \in \Phi_1$ and $\Psi_1\models C_i[M_1,\ldots, M_n]
  \bowtie N'_i$ for each $1\leq i \leq r$.
\end{itemize}

\noindent From this we easily deduce that $\Psi_1 \models N \bowtie f(N'_1,\ldots, N'_r)$.
We apply one $\mathbf{B}$ rule.
We have that $M_0 \ded t_0, N'_1 \ded u_1, \ldots,
N'_r \ded u_r \in \Phi_1$, $t = f(u_1,\ldots, u_r) \in \st(t_0)$ and $f \in
\F_{\pub}$. 
We distinguish two cases:

\smallskip{}

\noindent \emph{Rule \textbf{B.1}}: Assume that for all $M_t$ we have that $(M_t \ded t) \not\in
\Phi_1$. 

Let $\Phi'= \Phi_1 \cup \{f(N'_1,\ldots, N'_r) \ded t\}$, $\Psi'= \Psi_1$ and $N'=
f(N'_1,\ldots, N'_r)$. In order to conclude it remains to show that $\Psi'
\models N \bowtie N'$.
This is an easy consequence of the fact that
$\Psi_1 \models   N \bowtie f(N'_1,\ldots, N'_r)$.

\smallskip{}

\noindent \emph{Rule \textbf{B.2}}.
Assume that  there exists~$M_t$ such that $M_t \ded t \in \Phi_1$.

Let $\Phi'= \Phi_1$, $\Psi'= \Psi_1 \cup \{f(N'_1,\ldots, N'_r) \bowtie M_t\}$
and $N'=M_t$. In order to conclude it remains to show that $\Psi'\models N
\bowtie N'$. We have ${\Psi'\models f(N'_1,\ldots, N'_r) \bowtie N'}$ and
$\Psi'\models N \bowtie f(N'_1,\ldots, N'_r)$. This allows us to conclude.
\qed
\end{proof}

%%%%%%%%%%%%%%%%%%%%%%%%%%%%%%%%%%%%%%%%%%%%%%%%%%%%%%%%%%%%%%%%%%%%%%%%%%%%%%%

 \noindent\usebox{\lemcompsynteq}

\begin{proof}
By hypothesis, we have that $M \ded_{\Phi} t$ and $N \ded_{\Phi} t$ for some
term $t$.
By definition of $\ded_{\Phi}$, we have that
\begin{itemize}
\item  $M = C[M_1, \ldots, M_k]$, $N = C'[N_1, \ldots, N_\ell]$ for some contexts~${C,C'}$,
\item the facts  $M_1 \ded t_1,
\ldots, M_k \ded t_k$ and  $N_1 \ded u_1,\ldots, N_\ell \ded u_\ell$ are  in
$\Phi$,
\item $C[t_1, \ldots, t_k] = C'[u_1, \ldots, u_\ell]$.
\end{itemize}
We prove the result by structural induction on~$C$ and~$C'$. We assume
w.l.o.g. that $C$ is smaller than $C'$ (in terms of number of symbols).

\smallskip{}

\noindent \emph{Base case:} $C$ is reduced to an hole.
We have that ${C[M_1, \ldots, M_k] = M_1}$.
By hypothesis, we have that $N \ded_{\Phi} t = t_1$
 and thus $t \in \st(t_1)$. Thanks to Lemma~\ref{lem:compbasebase}, there exists
$(\Phi',\Psi')$ and~$N'$ such that
$(\Phi,\Psi) \Longrightarrow^* (\Phi',\Psi')$ using a \textbf{B} rule,
${N'\ded t_1 \in \Phi'}$ and $\Psi'\models N \bowtie N'$.
Since $M_1 \ded t_1$ and $N'\ded t_1$ are both in $\Phi'$, we deduce that
$N'=M_1$. Hence we have that $N' = M$ and thus we easily conclude.

\smallskip{}

\noindent \emph{Induction step:} $C =  f(C_1, \ldots, C_r)$ and $C'= f(C'_1, \ldots, C'_r)$
where $f \in \F_\pub$ is a symbol of arity~$r$ and $C_1, \ldots , C_r,
C'_1, \ldots, C'_r$ are contexts. Moreover, we have
that  $C_i[t_1, \ldots, t_k] = C'_i[u_1, \ldots,
  u_\ell]$ for every $1 \leq i \leq r$,
By applying the induction hypothesis, we deduce that there exists $(\Phi',
\Psi')$ such that 
\begin{itemize}
\item $(\Phi, \Psi) \Longrightarrow^* (\Phi', \Psi')$, and 
\item $\Psi'\models C_i[M_1, \ldots, M_k] \bowtie C'_i[N_1, \ldots,
  N_\ell]$ for every $1 \leq i \leq r$.
\end{itemize}
Hence, we have that $\Psi'\models M \bowtie N$. This allows us to conclude.
\qed
\end{proof}

%%%%%%%%%%%%%%%%%%%%%%%%%%%%%%%%%%%%%%%%%%%%%%%%%%%%%%%%%%%%%%%%%%%
%                                                                 % 
%                     Context reduction                           %
%                                                                 %
%%%%%%%%%%%%%%%%%%%%%%%%%%%%%%%%%%%%%%%%%%%%%%%%%%%%%%%%%%%%%%%%%%%

% \newcommand{\p}{\mathsf{p}}
%M.B. ah bon ? t|_p pose un souci ?
%S.D. c'est parce que le p 'simple' est utilise pour denoter le nombre de
%variable y

% lemme technique pour l'existence de la decomposition

The following lemma justifies the notion of decomposition
(Definition~\ref{def:decomp}) as far as completeness is concerned.

%\newsavebox{\lemdecomp}
%\sbox{\lemdecomp}{\vbox{%
\begin{lemma}[decomposition of a context reduction]\label{lem:decomp}
  Let $\Phi$ be a frame, $l$~a~(plain) term, $\sigma$ a
  substitution, and $M$ a term such that $M \ded_\Phi l\sigma$.
  Then there exist 
  \begin{itemize}
  \item a $(n,p,q)$-decomposition $D$ of $l$, written
    $l = D[l_1,\ldots,l_n,y_1,\ldots y_{p+q}]$,
  \item $n$ deduction facts $M_1 \ded t_1$, \ldots,  $M_n \ded t_n$ in $\Phi$,
  \item $p+q$ recipes $N_1$, \ldots, $N_{p+q}$
  \end{itemize}
  such that
  \begin{itemize}
  \item for every $1 \leq i \leq n$,\; $t_i = l_i \sigma$ and
  \item for every $1 \leq j \leq p+q$,\; $N_j \ded_\Phi y_j \sigma$.
  \end{itemize}
  In particular, $D[M_1,\ldots,M_n,N_1,\ldots N_{p+q}] \ded_\Phi l \sigma$.
  
  Besides, if $l$ is a left-hand side of rule in $\R$ and $\Phi$ is
  $\R$-reduced, $D$ is a proper decomposition (i.e. $D \neq \w_1$).
\end{lemma}
%}}

%\noindent\usebox{\lemdecomp}

%\noindent\usebox{\lemdecomp}

\begin{proof}
  Since $M \ded_\Phi l\sigma$, by definition there exists $C$ and
  $M^0_1 \ded t^0_1$, \ldots, $M^0_m \ded t^0_m$ in $\Phi$ such that
  $M=C[M^0_1,\ldots, M^0_m]$ and $l\sigma=C[t^0_1,\ldots, t^0_m]$.

  Let $x_1$, \ldots, $x_m$ be fresh variables. Given that
  $C[x_1,\ldots,x_m]$ and $l$ unify and have distinct variables, there
  exists a largest common context~$D_0$ such that
    $l = D_0[l^0_1,\ldots,l^0_{a}, y^0_1,\ldots, y^0_{b}]$ and 
    $C = D_0[\w_{j_1},\ldots,\w_{j_{a}}, D_1, \ldots, D_{b}]$
  where the terms $l^0_i$ are not variables and $D_0$ uses all his
  parameters: in particular % $y^0_k \in \var(l)$ and $\w_{j_{k}} \in
%   \{\w_1,\ldots,\w_m\}$, and
  $l\sigma=C[t^0_1,\ldots, t^0_m]$ means that
  \begin{itemize}
  \item for every $1 \leq k \leq a$,\; $l^0_k \sigma = t^0_{j_k}$, and
  \item for every $1 \leq k \leq b$,\; $y^0_k \sigma = D_k[t^0_1,\ldots, t^0_m]$
  \end{itemize}
  
  Let $n$ be the cardinal of $\{l^0_1,\ldots,l^0_{a}\}$. For each
  distinct $l_i$ in $\{l^0_1,\ldots,l^0_{a}\}$ ($1 \leq i \leq n$), we
  choose $k$ in $\{1,\ldots,a\}$ such that $l_i = l^0_k$ and define
  $M_i=M^0_k$ and $t_i=l^0_k \sigma=l_i\sigma$. Besides, for every $k'$ such
  that $l^0_{k'}= l^0_k$, we define $w_{k'} = \w_i$.

  Let $p$ be the cardinal of $\{y ^0_1,\ldots, y^0_{b}\} \cap
  \var(l_1,\ldots,l_n)$. For each distinct $y_j$ in $\{y ^0_1,\ldots,
  y^0_{b}\}$ ($1 \leq j \leq p$), we choose $k$ in $\{1,\ldots,b\}$
  such that $y_j = y^0_k$ and define $N_j = D_k[M^0_1,\ldots,
  M^0_m]$. Besides, for every $k'$ such that $y^0_{k'}= y^0_k$, we
  define $w_{a+k'} = \w_{p+j}$.

  Let $q=b-p$. We repeat the same operation for each distinct $y_j$ in
  $\{y ^0_1,\ldots, y^0_{b}\}-\var(l_1,\ldots,l_n)$ ($p+1 \leq j \leq p+q$).

  Finally, we let $D=D_0[w_1,\ldots,w_{a+b}]$. By construction, we
  have that
  \begin{itemize}
  \item $l = D[l_1,\ldots,l_n,y_1,\ldots y_{p+q}]$,
  \item the $l_i$ are mutually distinct non-variable terms and the
    $y_i$ are mutually distinct variables.
  \item $y_i \in \var(l_1, \ldots, l_n)$ iff $i \leq p$.
  \item $M_i \ded t_i$ is in $\Phi$,
  \item for every $1 \leq i \leq n$,\; $t_i = l_i \sigma$, and
  \item for every $1 \leq j \leq p+q$,\; $N_j \ded_\Phi y_j \sigma$.
  \end{itemize}
  
  \smallskip

  \noindent As for the last sentence, if $D$ is a parameter, so is
  $D_0$. As $l = y^0_k$ is impossible for a convergent system $\R$, we
  have $D_0 = \w_k$ with $k \leq a$.
  Hence $C = \w_{j_k}$ and $t^0_{k}=C[t^0_1,\ldots,t^0_k]=l\sigma$ is
  not $\R$-reduced. \qed
\end{proof}

%%%%%%%%%%%%%%%%%%%%%%%%%%%%%%%%%%%%%%%%%%%%%%%%%%%%%%%%%%%%%%%%%%%%%%%%%%%%%%%

\noindent\usebox{\lemcompred}
 
\begin{proof}
  By hypothesis, there exist a (public) context~$C$ and some
  deduction facts $M^0_1 \ded t^0_1$, \ldots, $M^0_{m_0} \ded
  t^0_{m_0} \in \Phi$ such that $M = C[M^0_1,\ldots, M^0_{m_0}]$ and
  $t = C[t^0_1,\ldots, t^0_{m_0}]$.

  Moreover, there exist a position~${\alpha}$, a substitution~$\sigma$ and a
  rewrite rule ${l \to r \in \R}$ such that $t|_{\alpha} = l\sigma$ and
  $t'= t[r\sigma]_{\alpha}$.

  We note that ${\alpha}$ must be a (symbol) position of $C$ since the
  $t_i^0$ are $\R$-reduced. Hence we may write $C|_{\alpha}[t^0_1,\ldots,
  t^0_{m_0}] = l\sigma$.

  By Lemma~\ref{lem:decomp}, we deduce that there exist 
  \begin{itemize}
  \item a proper $(n,p,q)$-decomposition $D$ of $l$\,: 
    $l = D[l_1,\ldots,l_n,y_1,\ldots y_{p}, z_1,\ldots z_{q}]$,
  \item $M_1 \ded t_1$, \ldots,  $M_n \ded t_n$ in $\Phi$,
  \item $N_1$, \ldots, $N_{p+q}$
  \end{itemize}
  such that
  \begin{itemize}
  \item for every $1 \leq i \leq n$, $t_i = l_i \sigma$,
  \item for every $1 \leq j \leq p$, $N_j \ded_\Phi y_j \sigma$, and
  \item for every $1 \leq k \leq q$, $N_{p+k} \ded_\Phi z_k \sigma$.
  \end{itemize}

  In particular, we obtain that
  \begin{eqnarray*}
    M|_{\alpha} = C|_{\alpha}[M^0_1,\ldots, M^0_{m_0}] &\ded_\Phi& C|_{\alpha}[t^0_1,\ldots, t^0_{m_0}] = l\sigma\\
    D[M_1,\ldots, M_n, N_1, \ldots, N_{p+q}] &\ded_\Phi& D[t_1,\ldots, t_n, y_1 \sigma, \ldots, y_p \sigma, z_1 \sigma, \ldots, z_q \sigma] = l\sigma
  \end{eqnarray*}
  
  Thus, by Lemma~\ref{lem:compbase}, there exists a derivation $(\Phi,
  \Psi) \Longrightarrow^* (\Phi_1, \Psi_1)$ using~\textbf{B}~rules such
  that $\Psi_1 \models M|_{\alpha} \qded D[M_1,\ldots, M_n, N_1, \ldots, N_{p+q}]$.

  Besides, since $y_j$ belongs to $\var(l_1,\ldots,l_n)$ by definition of
  decompositions, $y_j\sigma$ is a subterm of some
  $l_{i}\sigma=t_{i}$. Since $N_j \ded_\Phi y_j \sigma$, by
  applying Lemma~\ref{lem:compbasebase} repeatedly, we deduce that
  there exist some term $M_{n+1}$, \ldots, $M_{n+p}$ and a derivation
  $(\Phi_1, \Psi_1) \Longrightarrow^* (\Phi_2, \Psi_2)$ using
  \textbf{B}~rules such that for all $j$,
  \begin{itemize}
  \item $M_{n+j} \ded y_j \sigma$ is in $\Phi_2$, and
  \item $\Psi_2 \models M_{n+j} \qded N_j$.
  \end{itemize}

  Let $N = D[M_1,\ldots, M_{n+p}, N_{p+1}, \ldots, N_{p+q}]$.  We
  deduce that $N \ded_{\Phi_2} l\sigma$, and \[\Psi_2 \models M|_{\alpha}
  \qded D[M_1,\ldots, M_n, N_1, \ldots, N_{p+q}] \qded N\]

%  \smallskip{}

  \noindent We now consider the application to $(\Phi_2,\Psi_2)$ of a
  \textbf{A}~rule that involves the rewrite rule $l\to r$, the
  decomposition $D$, the plain terms
  $(t_1,\ldots,t_{n+p})=(l_1,\ldots,l_n, y_1,\ldots,y_p)\sigma$ and
  the substitution $\sigma'=\sigma|_V$ obtained by restricted the
  $\sigma$ to the domain $V=\var(l_1,\ldots,l_n)=\var(l_1,\ldots,l_n,
  y_1,\ldots,y_p)$.

%   Otherwise, rule~\textbf{A.1} or~\textbf{A.2} applies but may not be
%   strictly applicable.

  \smallskip{}

  \noindent \emph{Case} \textbf{A.3}. If $(r\sigma')\downR$ is not
  ground and $\Ctx(\Phi_2^+  \vdash^?_\R r\sigma') = \bot$ where $\Phi_2^+ = \Phi_2 \cup \{z_1 \ded
  z_1,\ldots, z_q \ded z_q\}$, then we may conclude that
  $(\Phi_2,\Psi_2) \Longrightarrow \bot$ by an instance of
  rule~\textbf{A.3} involving $l\to r$, the decomposition $D$ and the
  facts $M_1\ded t_1$,\ldots,$M_{n+p} \ded t_{n+p}$.

  \smallskip{}

  \noindent \emph{Case} \textbf{A.1}. If there exists $N_0 =
  \Ctx(\Phi_2^+ \vdash^?_\R r\sigma')$
  where $\Phi_2^+ = \Phi_2 \cup \{z_1 \ded z_1,\ldots, z_q \ded
  z_q\}$. By Property~(\ref{item:ctx_correct}) of $\Ctx$, let~$s_0$ be such
  that $N_0 \ded_{\Phi_2 \cup \{z_1,
    \ldots, z_q\}} s_0$ and $r\sigma' \rR^* s_0$, and define
  \begin{itemize}
  \item $\Phi' = \Phi_2$,
  \item $\Psi'= \Psi_2 \cup \{\forall z_1,\ldots, z_q. D[M_1,\ldots,
    M_{n+p},z_1,\ldots, z_q] \bowtie N_0\}$,
  \item $M'=M[M_0]_{\alpha}$ where $M_0= N_0 \;\{z_i \mapsto N_{p+i}\}_{1
      \leq i \leq q}$,
  \item $t''= t[t_0]_{\alpha} = t'[t_0]_{\alpha}$ where $t_0 = s_0 \;\{z_i \mapsto z_i
    \sigma\}_{1 \leq i \leq q}$.
  \end{itemize}
  By construction, we have $(\Phi_2,\Psi_2) \Longrightarrow
  (\Phi',\Psi')$ by an instance of rule \textbf{A.1}.

  Besides, $r\sigma' \rR^* s_0$ implies $t'|_{\alpha} = r\sigma \rR^*
  t_0$ and $t' \rR^* t''$.

  Given that $\alpha \in \pos(C)$ (where $C$ is the previously context related to
  $M \ded_{\Phi} t$) and $M_0 \ded_{\Phi'} t_0$, we have that
  $M'=M[M_0]_{\alpha}\ded_{\Phi'} t[t_0]_{\alpha} = t''$.

  It remains to show that $\Psi' \models M \bowtie M'$. Indeed, we
  have seen that $\Psi_2 \models M|_{\alpha} \qded N$ where $N =
  D[M_1,\ldots, M_{n+p}, z_1, \ldots, z_{q}]\{z_i \mapsto N_{p+i}\}_{1
    \leq i \leq q}$.
  Besides, by definition
  of $\Psi'$, it holds that $\Psi' \supseteq \Psi_2 \supseteq \Psi_1$ and we
  have that $\Psi'
  \models D[M_1,\ldots, M_{n+p}, z_1, \ldots, z_{q}] \qded
  N_0$. Therefore, $\Psi' \models M|_{\alpha} \qded M_0$ and $\Psi'
  \models M \bowtie M[M_0]_{\alpha}=M'$.

  \smallskip
  \noindent \emph{Case} \textbf{A.2}: if $(r\sigma')\downR$ is ground
  and $\Ctx(\Phi_2^+\vdash^?_\R
  r\sigma') = \bot$ where $\Phi_2^+ = \Phi_2 \cup
  \{z_1 \ded z_1,\ldots, z_q \ded z_q\}$, define
  \begin{itemize}
  \item $M_0 = D[M_1,\ldots, M_{n+p}, \cst a, \ldots, \cst a]$ and
    $t_0=(r\sigma')\downR$,
  \item $\Phi' = \Phi_2 \cup \{M_0 \ded t_0\}$,
  \item $\Psi'= \Psi_2 \cup\, \{\forall z_1,\ldots, z_q. D[M_1,\ldots,
    M_{n+p},z_1,\ldots, z_q] \bowtie M_0\}$,
  \item $M'= M[M_0]_\alpha$, and
  \item $t''= t[t_0]_\alpha$.
  \end{itemize}
  where~$\cst a$ is the fixed public constant of rule \textbf{A.2}.

  By construction, $(\Phi,\Psi) \Longrightarrow (\Phi',\Psi')$ by an
  instance of the \textbf{A.2} rule.

  Since $t_0$ is ground and $\sigma=\sigma'\sigma$, we
  have $t_0 = (r\sigma)\downR$. Therefore $t' = t[r\sigma]_\alpha
  \rR^* t[\,(r\sigma)\downR\,]_\alpha = t''$.

  Given that $\alpha \in
  \pos(C)$ %(where $C$ is the context showing that $M \ded_{\Phi} t$)
  and by construction $M_0 \ded_{\Phi'} t_0$, we have
  $M'\ded_{\Phi'} t''$.

  It remains to show that $\Psi' \models M \bowtie M'$. Indeed, we
  have seen that $\Psi_2 \models M|_{\alpha} \qded N$ where $N =
  D[M_1,\ldots, M_{n+p}, z_1, \ldots, z_{q}]\{z_i \mapsto N_{p+i}\}_{1
    \leq i \leq q}$.
  By definition of $\Psi'$, it holds that $\Psi' \models N \qded M_0$
  hence $\Psi' \models M \bowtie M[N]_{\alpha} \bowtie M[M_0]_{\alpha} = M'$.

  \smallskip

  \noindent
  The additional properties claimed on the derivation are clear from
  the construction above. \qed
\end{proof}

\end{document}